\documentclass[journal]{IEEEtran}

\ifx\pdfoutput\undefined
\usepackage{graphicx} \else
\usepackage[pdftex]{graphicx} \fi
\usepackage{subfigure}
\usepackage{amsmath}
\usepackage{dsfont}
\usepackage{amssymb}
\usepackage{algorithm}
\usepackage{subfigure}
\usepackage{cuted}
\usepackage{epstopdf}
\usepackage{cite}
\usepackage{soul,color,bm}
\usepackage{setspace}
\usepackage{amsbsy}

\newtheorem{pro}{Proposition}
\newtheorem{remark}{Remark}

\usepackage{setspace}
\usepackage{lipsum} 
\usepackage{multicol}
\usepackage{float}
\usepackage{graphicx}

\usepackage{etoolbox}
\makeatletter
\patchcmd{\@begintheorem}{\textit}{\mathbf}{}{}
\makeatother
\ifCLASSINFOpdf
\else
\fi

\begin{document}

\title{Rethinking Joint UAV Placement and Beamforming: A Correlation-Aware Geometric Approach}

\author{Chaeyeon Kim and Kisong Lee,~\IEEEmembership{Senior Member,~IEEE}


\thanks{The authors are with the Department of Information and Communication Engineering, Dongguk University, Seoul 04620, South Korea (e-mail: kslee851105@gmail.com).}

}

\maketitle

\begin{abstract}
In multiuser unmanned aerial vehicle (UAV)-assisted downlink communications, UAV placement and transmit beamforming are inherently coupled through the propagation geometry. However, fully joint design based on instantaneous channel state information (CSI) is impractical, because the small-scale fading depends on the UAV location to be optimized and thus is unavailable a priori. Moreover, existing joint placement and beamforming methods do not explicitly optimize the UAV position with respect to the geometry-dependent multiuser interference induced by inter-user steering correlation. To address this issue, we propose a correlation-aware geometric framework for joint UAV placement and beamforming. Specifically, the UAV position is first optimized based on long-term channel statistics, where the steering-vector correlation is incorporated into the placement design through a conservative Gaussian surrogate that avoids interference underestimation. The resulting nonconvex positioning problem is then handled using successive convex approximation, auxiliary-variable decoupling, and quadratic transform techniques. For the obtained UAV location, the transmit beamformer is then optimized using instantaneous CSI. Simulation results show that the proposed framework significantly improves the minimum user spectral efficiency by enhancing angular separability among users and reducing inter-user interference. These results demonstrate that UAV placement should be designed not only for desired-link enhancement but also for interference mitigation through geometry-aware user separation.
\end{abstract}


\IEEEpeerreviewmaketitle

\vspace{-2mm}
\section{Introduction}

Unmanned aerial vehicles (UAVs) have emerged as a promising platform for next-generation wireless communications due to their high mobility, flexible deployment, and favorable air-to-ground propagation characteristics \cite{Geraci22,Zhou23}. Unlike conventional terrestrial infrastructure, UAVs can be rapidly dispatched and repositioned according to traffic demand, service priority, or environmental changes, making them particularly attractive for temporary hotspot coverage, emergency connectivity, data collection, and relay-assisted communications \cite{Lin2018}. More importantly, the three-dimensional (3D) mobility of UAVs enables active control over communication geometry and propagation conditions, so that the UAV location becomes not merely an operational parameter but a fundamental design variable that directly affects link quality, interference, and network performance. This unique controllability has made UAV deployment and trajectory optimization central research problems in UAV-enabled wireless systems.

Motivated by these advantages, prior work on UAV-enabled wireless communications has investigated a wide range of design problems, including deployment optimization, trajectory design, user association, and radio resource allocation. In many of these studies, UAV mobility is primarily exploited to improve large-scale propagation conditions, such as path-loss reduction, line-of-sight (LoS) enhancement, and coverage extension. Accordingly, UAV positioning and trajectory design have often been developed based on distance- or elevation-angle-dependent channel models, forming an important foundation for UAV communication design. In particular, \cite{Mozaffari16,Fan18} considered UAV placement optimization with the UAV serving as a quasi-static base station to enlarge the service coverage region. From the spectral-efficiency (SE) perspective, \cite{You20,Heo2024-2} optimized UAV trajectories and communication resources to improve throughput or data-harvesting performance. From the energy-efficiency (EE) perspective, UAV mobility was exploited to enhance the overall EE of UAV-enabled communications through joint trajectory and transmission optimization while accounting for propulsion-related energy consumption and, in some cases, energy-harvesting considerations \cite{Zeng17,Zeng19,Yang20}. Related efforts also considered wireless-powered networks, in which UAVs deliver wireless energy to ground nodes (GNs) with energy-harvesting capability while jointly coordinating power transfer and information transmission \cite{Heo24,Park24,Park26}. In addition, multi-UAV systems were investigated to improve coverage and scalability, with particular emphasis on interference-aware trajectory and resource design \cite{Wu18,Kim24,Kim25}. UAV maneuverability was also exploited for secure communications, including anti-jamming trajectory design and cooperative-jamming-based secrecy enhancement \cite{Lee18,Duo20,Heo2024-3}.

However, as UAV-enabled wireless networks increasingly adopt multi-antenna transmission and spatially selective communications, large-scale-channel-oriented designs are often insufficient to fully capture the role of UAV positioning. In particular, the UAV location affects not only the link strength but also the spatial structure of the multiuser channel, including the angular separability among users, the steering-vector correlation, and the resulting inter-user interference pattern. Consequently, UAV positioning should be designed not only for desired-link enhancement but also for interference management through propagation geometry. Motivated by this broader role of UAV positioning, several recent studies have investigated joint UAV positioning and beamforming in various settings, although mostly in forms that do not explicitly capture the geometry-dependent interference structure induced by inter-user steering correlation. For example, \cite{Won24} studied the joint optimization of aerial base-station location, beam direction, beamwidth, and radio-resource allocation for sum-rate maximization in UAV networks with controllable directional antennas. In the context of mmWave UAV communications, \cite{Yu22} investigated the joint design of UAV positioning, beamforming, and power allocation for EE maximization in multiuser UAV-aided systems. In addition, \cite{Zhu22} considered multi-UAV networks and jointly optimized UAV positioning, user clustering, and hybrid beamforming to maximize the achievable sum rate under minimum-rate constraints, while \cite{Liu23} studied UAV deployment and robust hybrid beamforming under UAV jitter for max--min rate enhancement. In multiuser relay networks, \cite{Chou25} addressed the joint design of UAV relay placement, beamforming, and receive combining to improve the minimum user signal-to-interference-plus-noise ratio (SINR). Furthermore, \cite{Xu25} investigated joint UAV placement and beamforming in a multistage integrated sensing and communication framework with wide-beam sensing and narrow-beam localization. 

Despite this progress, important limitations remain. In particular, the designs in \cite{Yu22,Zhu22,Liu23,Chou25} optimize UAV deployment mainly from the perspective of desired-signal enhancement, without incorporating inter-user steering correlation into the placement stage. As a result, the UAV position is not directly optimized for the resulting interference structure. Meanwhile, the framework in \cite{Xu25} refines UAV placement via first-order Taylor-based successive convex approximation (SCA) applied to angle-dependent channel expressions. Although this captures angle variations caused by UAV movement, it relies on local linearization and does not explicitly model the steering-correlation structure governing multiuser interference. Motivated by this gap, we reconsider joint UAV placement and beamforming for multiuser downlink communications by incorporating steering-vector correlation into the placement design. This enables the UAV location to be optimized not only for desired-link enhancement but also for interference mitigation through improved angular user separability. The transmit beamformer is then optimized for the obtained UAV position using instantaneous channel state information (CSI) fed back from the GNs. The contributions of our work are summarized as follows:
\begin{itemize}
    \item We identify a key limitation of existing joint UAV placement and beamforming designs: prior approaches do not explicitly incorporate the steering-vector correlation structure governing geometry-dependent multiuser interference into the placement stage; instead, they either optimize UAV placement primarily for desired-signal enhancement or treat angular variations only through local linearization \cite{Yu22,Zhu22,Liu23,Chou25,Xu25}. To the best of our knowledge, this work is the first to embed this structure into UAV positioning, thereby optimizing the UAV location not only for desired-link enhancement but also for interference mitigation through improved angular user separability.

    \item Building on this viewpoint, we develop a tractable two-stage optimization framework. For the positioning stage based on long-term channel statistics, we construct a safeguarded Gaussian surrogate that preserves the dominant peak structure of the steering correlation while avoiding sidelobe underestimation, and use it to derive a conservative SE formulation. The resulting nonconvex positioning problem is then handled using SCA, auxiliary-variable decoupling, and quadratic transform. For the obtained UAV location, the beamformer is optimized using instantaneous CSI. The overall algorithm admits convergent iterative updates with polynomial-time complexity.

    \item Simulation results show that the proposed framework significantly improves the minimum user SE by enlarging angular separation among users and thereby reducing inter-user interference. More importantly, the results reveal that, in multiuser downlink systems, the quality of UAV placement depends not only on link-strength improvement but also on its ability to create a more spatially separable multiuser channel.
\end{itemize}

The remainder of this paper is organized as follows. Section II introduces the system model and formulates the joint UAV placement and beamforming problem for multiuser downlink communications. Section III incorporates the structured steering-vector correlation into the UAV placement problem and proposes a two-stage optimization framework for solving the joint design problem. Section IV presents simulation results and discussions, and Section V concludes the paper.

\vspace{-2mm}
\section{System Model and Problem Statement}

\begin{figure}[t]
    \centerline{\includegraphics[width=0.8\linewidth]{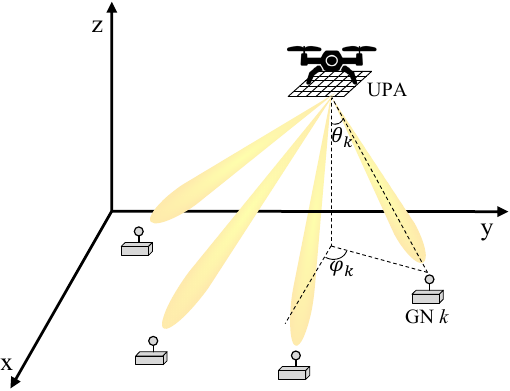}} \vspace{-2mm}
    \caption{A multiuser UAV-assisted communication system.} 
    \label{fig1} \vspace{-4mm}
\end{figure}

As shown in Fig. \ref{fig1}, we consider a multiuser UAV-assisted downlink communication system, where a UAV equipped with a uniform planar array (UPA) consisting of $N_x \!\times\! N_y$ antennas serves $K$ single-antenna GNs, indexed by $k\!\in\!\mathcal{K}\!=\!\{1,\dots,K\}$. The 3D location of the UAV is denoted by $\mathbf{q}\!=\![x,y,z]^T$, while the position of GN $k$ is given by $\mathbf{u}_k\!=\![x_k,y_k,z_k]^T$. 

Due to the high likelihood of LoS links in UAV communications, the air-to-ground channel is modeled as a Rician fading channel \cite{Chou25}. Specifically, the downlink channel vector from the UAV to GN $k$ is given by
\begin{align}
\mathbf{h}_k=\sqrt{\frac{\beta}{\|\mathbf{q}\!-\!\mathbf{u}_k\|^{\alpha}}} \Bigg( \sqrt{\frac{K_R}{K_R\!+\!1}}\mathbf{h}^{\textrm{L}}_k\!+\!\sqrt{\frac{1}{K_R\!+\!1}}\mathbf{h}^{\textrm{N}}_k\Bigg), \label{cch}
\end{align}
where $\alpha$ and $\beta$ denote the path-loss exponent and the reference channel power gain, respectively, $\|\mathbf{q}\!-\!\mathbf{u}_k\|$ represents the Euclidean distance between the UAV and GN $k$, and $K_R$ denotes the Rician $K$-factor in linear scale. Moreover, the non-line-of-sight (NLoS) component $\mathbf{h}_k^{\textrm{N}} \!\in\! \mathbb{C}^{N_xN_y\times 1}$ captures the small-scale fading caused by multipath scattering, and its elements are assumed to be independent and identically distributed circularly symmetric complex Gaussian random variables with zero mean and unit variance, i.e., $\mathbf{h}_k^{\textrm{N}} \!\sim\! \mathcal{CN}(\mathbf{0},\mathbf{I})$.

The LoS component $\mathbf{h}_k^{\textrm{L}}\!\in\! \mathbb{C}^{N_xN_y\times 1}$ represents the deterministic array response determined by the relative 3D geometry between the UAV and GN $k$. In particular, it depends on the angular direction from the UAV to GN $k$, characterized by the elevation angle $\theta_k$ and azimuth angle $\phi_k$, where $\theta_k$ is the angle between the direction of GN $k$ and the UAV's downward vertical axis, and $\phi_k$ is the corresponding azimuth angle on the horizontal plane. For the considered UPA, the array response is expressed as the Kronecker product of two uniform linear array (ULA) responses along the $x$- and $y$-axes, i.e.,
\begin{align}
\mathbf{h}_k^{\textrm{L}} = \mathbf{a}_k(\theta_k,\phi_k)
=\mathbf{a}_k^x(\theta_k,\phi_k) \otimes \mathbf{a}_k^y(\theta_k,\phi_k),
\end{align}
where $\mathbf{a}_k^x(\theta_k,\phi_k)$ and $\mathbf{a}_k^y(\theta_k,\phi_k)$ are given by
\begin{align}
\mathbf{a}_k^x(\theta_k,\phi_k) &= \left[1,e^{-j\frac{2\pi d}{\lambda}\Omega_k^x},\dots,e^{-j\frac{2\pi d(N_x-1)}{\lambda}\Omega_k^x}\right]^T, \\
\mathbf{a}_k^y(\theta_k,\phi_k) &= \left[1,e^{-j\frac{2\pi d}{\lambda}\Omega_k^y},\dots,e^{-j\frac{2\pi d(N_y-1)}{\lambda}\Omega_k^y}\right]^T.
\end{align}
Here, $\Omega_k^x\!=\!\sin\theta_k\cos\phi_k\!=\!\frac{x_k-x}{\|\mathbf{q}-\mathbf{u}_k\|}$ and 
$\Omega_k^y\!=\!\sin\theta_k\sin\phi_k\!=\!\frac{y_k-y}{\|\mathbf{q}-\mathbf{u}_k\|}$ denote the direction cosines along the $x$- and $y$-axes, respectively. In addition, $\lambda\!=\!\frac{c}{f_c}$ is the carrier wavelength, where $c$ is the speed of light and $f_c$ is the carrier frequency, while $d$ represents the antenna spacing.

Let $\mathbf{s}\!=\![s_1,\dots,s_K]^T \!\in\! \mathbb{C}^{K \times 1}$ denote the data symbol vector intended for the $K$ GNs, where $\mathbb{E}\{\mathbf{s}\}\!=\!\mathbf{0}$ and $\mathbb{E}\{\mathbf{s}\mathbf{s}^H\}\!=\!\mathbf{I}$. The UAV employs linear precoding through the beamforming matrix $\mathbf{W} \!=\! [\mathbf{w}_1,\dots,\mathbf{w}_K] \!\in\! \mathbb{C}^{N_xN_y \times K}$, where $\mathbf{w}_k$ is the beamforming vector for GN $k$. Accordingly, the transmitted signal is given by $\mathbf{x} \!=\! \mathbf{W}\mathbf{s}$. The received signal at GN $k$ is then expressed as
\begin{align}
y_{k} &= \mathbf{h}_k^H \mathbf{x} + n_k \nonumber \\
&= \mathbf{h}_k^H \mathbf{w}_k s_k 
+ \sum_{i\neq k} \mathbf{h}_k^H \mathbf{w}_i s_i
+ n_k, \label{yk}
\end{align}
where $n_k \!\sim\! \mathcal{CN}(0,\sigma^2)$ denotes the additive white Gaussian noise (AWGN) at GN $k$.

In \eqref{yk}, the first term represents the desired signal intended for GN $k$, while the second term accounts for the multiuser interference caused by signals intended for the other GNs. Accordingly, the resulting achievable SE at GN $k$ is given by
\begin{align}
    R_{k}=\log_2\bigg(1+\frac{|\mathbf{h}^{H}_{k}\mathbf{w}_k|^2}{\sum_{i \neq k}|\mathbf{h}^{H}_{k}\mathbf{w}_i|^2+\sigma^2}\bigg). \label{Rk}
\end{align}

Our objective is to jointly optimize the UAV position $\mathbf{q}$ and the beamforming matrix $\mathbf{W}$ to maximize the minimum SE among all GNs under the max--min fairness criterion. By introducing an auxiliary variable $\eta$ to represent the minimum achievable SE, the optimization problem is formulated as
\begin{align} 
\textbf{(P0):} ~\max_{\mathbf{W},~\mathbf{q},~\eta} & ~~~~~~~~ \eta \nonumber  \\
\textrm{subject to} & ~~~~ R_{k} \geq \eta, ~\forall k, \label{constrk} \\
& ~~~~\textrm{tr}(\mathbf{W}\mathbf{W}^H) \leq P_{\textrm{max}}, \label{pmax} \\ 
& ~~~~ \mathbf{q} \in \mathcal{Q}. \label{qmax}
\end{align}
Constraint \eqref{pmax} imposes a total transmit power limit at the UAV, where $\textrm{tr}(\mathbf{W}\mathbf{W}^H)\!=\!\sum_{k=1}^{K}\|\mathbf{w}_k\|^2$ denotes the total transmit power across all antennas and $P_{\textrm{max}}$ is the maximum allowable transmit power. Constraint \eqref{qmax} confines the feasible 3D operating region of the UAV, where $\mathcal{Q} \!\triangleq\! \{\mathbf{q} \!\in\! \mathbb{R}^{3} \big|x_{\textrm{min}}\!\le\! x \!\le\! x_{\textrm{max}},y_{\textrm{min}} \!\le\! y \!\le\! y_{\textrm{max}},z_{\textrm{min}} \!\le\! z \!\le\! z_{\textrm{max}}\}$ with $x_{\textrm{min}},y_{\textrm{min}},z_{\textrm{min}} \!\ge\! 0$.

\vspace{-2mm}
\section{Proposed Optimization Framework}

We propose a two-stage optimization framework that separates UAV positioning and beamforming design. Fully joint optimization based on instantaneous CSI is generally impractical, because the small-scale fading depends on the UAV position to be optimized and therefore cannot be known a priori during the positioning stage. Accordingly, in the first stage, the UAV position is optimized based on long-term channel statistics. In the second stage, once the UAV position is fixed, instantaneous CSI is acquired through feedback, and the beamforming matrix is optimized to further mitigate residual interference and improve system performance.

\vspace{-2mm}
\subsection{UAV Positioning Based on Long-Term Channel Statistics}

In the first stage, we optimize the UAV position based on long-term channel statistics. To this end, maximum ratio transmission (MRT) is adopted as the reference beamforming scheme. Since MRT maximizes each user’s desired signal along its own steering direction without actively nulling inter-user interference, the residual interference is determined largely by the similarity among user steering directions, as quantified by the steering-vector correlation. Consequently, UAV positioning affects the resulting SE not only through path loss but also through steering-vector correlation.

\subsubsection{Steering-Correlation-Aware Surrogate Construction}

Under the MRT-based positioning criterion, the mean channel of GN $k$ is given by
\begin{align}
    \bar{\mathbf{h}}_{k} \triangleq \mathbb{E}[\mathbf{h}_{k}]=\sqrt{\frac{\beta}{\|\mathbf{q}-\mathbf{u}_k\|^{\alpha}}} \sqrt{\frac{K_R}{K_R+1}} \mathbf{h}^{\textrm{L}}_k.
\end{align}
Accordingly, the MRT beamforming vector for GN $k$ is aligned with $\bar{\mathbf{h}}_{k}$ and is represented by
\begin{align}
\mathbf{w}_k 
= \sqrt{p_k}\frac{\bar{\mathbf{h}}_{k}}{\|\bar{\mathbf{h}}_{k}\|}
= \sqrt{p_k}\frac{\mathbf{a}_k(\theta_k,\phi_k)}{\|\mathbf{a}_k(\theta_k,\phi_k)\|}, ~\forall k,
\label{eq:MRT_wk}
\end{align}
where $p_k$ denotes the transmit power allocated to GN $k$. 

Under this reference beamforming choice, the resulting long-term performance depends explicitly on both the desired array gain and the inter-user steering-vector correlation. Hence, the positioning design can exploit not only the distance-dependent path loss but also the angular separation among GNs. To make this dependence explicit, we characterize the average desired signal power and the average inter-user interference power under the MRT beamforming vectors designed from the mean channel vectors. Specifically, for GN $k$, these terms are expressed as
\begin{align}
\mathbb{E}[|{\mathbf{h}}^{H}_{k}\mathbf{w}_k|^2]&=\mathbf{w}^{H}_{k}\pmb{\Phi}_k\mathbf{w}_{k}, \\
\mathbb{E}[|{\mathbf{h}}^{H}_{k}\mathbf{w}_i|^2]&=\mathbf{w}^{H}_{i}\pmb{\Phi}_k\mathbf{w}_{i}, ~\forall i \neq k,
\end{align}
where $\pmb{\Phi}_k$ is the channel correlation matrix represented by
\begin{align}
\pmb{\Phi}_k \!=\! \frac{\beta}{\|\mathbf q\!-\!\mathbf u_k\|^\alpha}\!
\left(\!\frac{K_R}{K_R\!+\!1}\mathbf{a}_k(\theta_k,\phi_k)\mathbf{a}^H_k(\theta_k,\phi_k) \!+\! \frac{1}{K_R\!+\!1} \mathbf{I} \!\right).
\label{eq:Rk_identity}
\end{align}
Using these average signal and interference power terms, we adopt the following long-term geometry-dependent SE approximation for GN $k$ \cite{Gan21}:
\begin{align}
    R_{k}
    =\log_2 \left(1+\frac{\mathbf{w}^{H}_{k}\pmb{\Phi}_k\mathbf{w}_{k}}
    {\sum_{i \neq k}\mathbf{w}^{H}_{i}\pmb{\Phi}_k\mathbf{w}_{i}+\sigma^2} \right).
    \label{rkkk}
\end{align}

To make the geometry dependence under MRT more explicit, we derive closed-form expressions for the desired signal term $\mathbf{w}^{H}_{k}\pmb{\Phi}_k\mathbf{w}_{k}$ and the inter-user interference terms $\{\mathbf{w}^{H}_{i}\pmb{\Phi}_k\mathbf{w}_{i},\forall i\!\neq\! k\}$. First, the desired signal term is given by
\begin{align}
\mathbf{w}^{H}_{k}\pmb{\Phi}_k\mathbf{w}_{k}
&=\frac{p_k}{\|\mathbf a_k(\theta_k,\phi_k)\|^2}\mathbf{a}^H_k(\theta_k,\phi_k)\pmb{\Phi}_k\mathbf{a}_k(\theta_k,\phi_k)\nonumber\\
&=\frac{{\beta}p_k}{\|\mathbf{q}\!-\!\mathbf{u}_k\|^{\alpha}}
\bigg(\!\frac{K_R}{K_R\!+\!1}\big\|\mathbf{a}_k(\theta_k,\phi_k)\big\|^{2}\!+\!\frac{1}{K_R\!+\!1}\!\bigg) \nonumber\\
&=\frac{{\beta}p_k}{\|\mathbf{q}\!-\!\mathbf{u}_k\|^{\alpha}}\bigg(\!\frac{K_R}{K_R\!+\!1}N_xN_y\!+\!\frac{1}{K_R\!+\!1}\!\bigg),
\end{align}
where $\|\mathbf{a}_k(\theta_k,\phi_k)\|^{2}\!=\!N_xN_y,\forall k$ for an $N_x \!\times\! N_y$ UPA. Hence, under MRT, the desired term depends on the distance-dependent path loss and a constant array-gain-related factor.

Similarly, the inter-user interference term contributed by beam $i$ to GN $k$ is given by
\begin{align}
&\mathbf{w}^{H}_{i}\pmb{\Phi}_k\mathbf{w}_{i}
=\frac{p_i}{\|\mathbf a_i(\theta_i,\phi_i)\|^2}\mathbf{a}^H_i(\theta_i,\phi_i)\pmb{\Phi}_k\mathbf{a}_i(\theta_i,\phi_i) \nonumber\\
&=\frac{{\beta}p_i}{\|\mathbf{q}\!-\!\mathbf{u}_k\|^{\alpha}}
\bigg(\!\frac{K_R}{K_R\!+\!1}\frac{\left|\mathbf{a}_k^{H}(\theta_k,\phi_k)\mathbf{a}_i(\theta_i,\phi_i)\right|^{2}}{\|\mathbf{a}_i(\theta_i,\phi_i)\|^{2}}\!+\!\frac{1}{K_R\!+\!1}\!\bigg) \nonumber\\
&=\frac{{\beta}p_i}{\|\mathbf{q}\!-\!\mathbf{u}_k\|^{\alpha}}\bigg(\!\frac{K_R}{K_R\!+\!1}\frac{\left|\mathbf{a}_k^{H}(\theta_k,\phi_k)\mathbf{a}_i(\theta_i,\phi_i)\right|^{2}}{N_xN_y}\!+\!\frac{1}{K_R+1}\!\bigg). \label{inter}
\end{align}
Unlike the desired term, the interference term explicitly depends on the inter-user steering-vector correlation and thus reflects the angular separation among GNs.

The correlation term $\left|\mathbf{a}_k^{H}(\theta_k,\phi_k)\mathbf{a}_i(\theta_i,\phi_i)\right|^{2}$ in \eqref{inter} can be written as the product of two Dirichlet-kernel magnitude-squared terms \cite{Zhang23}:
\begin{align}
&\left|\mathbf{a}_k^{H}(\theta_k,\phi_k)\mathbf{a}_i(\theta_i,\phi_i)\right|^{2}
\!=\!\left|\sum_{n_y=0}^{N_y-1}\sum_{n_x=0}^{N_x-1}\!e^{-j \psi^x_{k,i} n_x } e^{-j \psi^y_{k,i} n_y}\right|^{2} \nonumber\\
&~~~~~~~~~~~~~~~~~~~~~=\!\left|\sum_{n_x=0}^{N_x-1} \!e^{-j \psi^x_{k,i} n_x}\right|^{2}
\left|\sum_{n_y=0}^{N_y-1} \!e^{-j \psi^y_{k,i} n_y}\right|^{2} \nonumber\\
&~~~~~~~~~~~~~~~~~~~~~=\!\left|\frac{\sin\!\left(\frac{N_x\psi^x_{k,i}}{2}\right)}{\sin\!\left(\frac{\psi^x_{k,i}}{2}\right)}\right|^{2}
\left|\frac{\sin\!\left(\frac{N_y\psi^y_{k,i}}{2}\right)}{\sin\!\left(\frac{\psi^y_{k,i}}{2}\right)}\right|^{2}, \label{dirichlet_prod}
\end{align}
where $\psi^x_{k,i}\!=\!\frac{2\pi d}{\lambda}(\Omega^x_k\!-\!\Omega^x_i)$ and $\psi^y_{k,i}\!=\!\frac{2\pi d}{\lambda}(\Omega^y_k\!-\!\Omega^y_i)$. 

Note that \eqref{dirichlet_prod} is periodic, with a sharp main peak around $(\psi^x_{k,i},\psi^y_{k,i})\!=\!(0,0)$ and sidelobes away from the peak. This behavior reflects the fact that larger arrays produce narrower main lobes and hence stronger angular selectivity. Since the MRT-induced interference term explicitly depends on $\left|\mathbf{a}_k^{H}(\theta_k,\phi_k)\mathbf{a}_i(\theta_i,\phi_i)\right|^{2}$, increasing the angular separation among GNs generally reduces $\left|\mathbf{a}_k^{H}(\theta_k,\phi_k)\mathbf{a}_i(\theta_i,\phi_i)\right|^{2}$ and thereby suppresses inter-user interference. This observation motivates geometry-aware UAV positioning.

\begin{pro}[Gaussian approximation of the Dirichlet kernel near the main lobe]
Let $D_N(\psi) \!\triangleq\! \frac{\sin\left(\frac{N\psi}{2}\right)}{\sin\left(\frac{\psi}{2}\right)}$. Then, for sufficiently small $|\psi|$, the squared magnitude of $D_N(\psi)$ can be approximated as
\begin{align}
\left|D_N(\psi)\right|^2=\left|\frac{\sin\!\left(\frac{N\psi}{2}\right)}{\sin\!\left(\frac{\psi}{2}\right)}\right|^2 \approx N^2 e^{-\frac{(N^2-1)}{12}\psi^2}.
\end{align}
\end{pro}

\begin{IEEEproof}
Since the main lobe is centered at $\psi\!=\!0$, we expand $\left|D_N(\psi)\right|^2$ around $\psi\!=\!0$. Taking the logarithm gives
\begin{align}
    \ln \left|D_N(\psi)\right|^2
    =
    2\ln\left|\sin\!\left(\frac{N\psi}{2}\right)\right|
    -
    2\ln\left|\sin\!\left(\frac{\psi}{2}\right)\right|.
\end{align}
Using the Taylor expansion of $\ln\left|\frac{\sin x}{x}\right|$ around $x\!=\!0$, i.e., $\ln\left|\frac{\sin x}{x}\right|\!=\!-\frac{x^2}{6}\!-\!\frac{x^4}{180}\!+\!O(x^6)$, we obtain
\begin{align}
\ln |\sin x| = \ln|x|-\frac{x^2}{6}-\frac{x^4}{180}+O(x^6).
\end{align}
Applying this expansion to $x\!=\!\frac{N\psi}{2}$ and $x\!=\!\frac{\psi}{2}$ yields
\begin{align}
\!\!\!2\ln\!\left|\sin\!\left(\frac{N\psi}{2}\right)\right| &=\! 2\ln\!\left|\frac{N\psi}{2}\right|\!-\!\frac{N^2\psi^2}{12} \!-\!\frac{N^4\psi^4}{720} \!+\!O(\psi^6),\\
\!\!\!2\ln\!\left|\sin\!\left(\frac{\psi}{2}\right)\right| &=\! 2\ln\!\left|\frac{\psi}{2}\right|\!-\!\frac{\psi^2}{12}\!-\!\frac{\psi^4}{720}\!+\!O(\psi^6).
\end{align}
Taking the difference, we obtain
\begin{align}
\ln \!\left|D_N(\psi)\right|^2 \!=\! 2\ln \!N \!-\!\frac{N^2\!-\!1}{12}\psi^2\!-\!\frac{N^4\!-\!1}{720}\psi^4\!+\!O(\psi^6).
\end{align}
For sufficiently small $|\psi|$, neglecting the fourth- and higher-order terms gives
\begin{align}
\ln \left|D_N(\psi)\right|^2 \approx 2\ln N-\frac{N^2-1}{12}\psi^2.
\end{align}
Exponentiating both sides yields
\begin{align}
\left|D_N(\psi)\right|^2 \approx N^2 e^{-\frac{(N^2-1)}{12}\psi^2},
\end{align}
which completes the proof.
\end{IEEEproof}

Using Proposition 1, the oscillatory product form in \eqref{dirichlet_prod} can be approximated around the main-lobe region $(\psi^x_{k,i},\psi^y_{k,i})\!\approx\!(0,0)$ by the following smooth Gaussian surrogate:
\begin{align}
&\left|\mathbf{a}_k^{H}(\theta_k,\phi_k)\mathbf{a}_i(\theta_i,\phi_i)\right|^{2}
\underset{(\psi^x_{k,i},\psi^y_{k,i})\approx(0,0)}{\approx} \nonumber \\
&~~~~~~~~~~\bigg(N_x^2 e^{-\frac{(N_x^2-1)}{12}(\psi^x_{k,i})^{2}}\bigg)
\bigg(N_y^2 e^{-\frac{(N_y^2-1)}{12}(\psi^y_{k,i})^{2}}\bigg).
\label{eq:dirichlet_gauss}
\end{align}
While \eqref{eq:dirichlet_gauss} provides a smooth and locally accurate approximation near the main lobe, the exact Dirichlet-kernel product in \eqref{dirichlet_prod} is $2\pi$-periodic in each phase variable. In particular, under half-wavelength spacing $d\!=\!\frac{\lambda}{2}$, we have $\psi^\nu_{k,i}\!=\!\frac{2\pi d}{\lambda}(\Omega_k^\nu\!-\!\Omega_i^\nu)\!\in\![-2\pi,2\pi]$ for $\nu\!\in\!\{x,y\}$, since $\Omega_k^\nu,\Omega_i^\nu\!\in\![-1,1]$. Over this interval, the periodic structure gives rise not only to the dominant peak around $\psi^\nu_{k,i}\!\approx\! 0$ but also to two periodic peaks near the boundaries $\psi^\nu_{k,i}\!\approx\! \pm 2\pi$. Therefore, a single Gaussian centered at the origin may underestimate the correlation when $\psi^\nu_{k,i}$ is close to $\pm 2\pi$.

To capture these three dominant peaks while retaining a smooth and tractable form, we approximate each one-dimensional Dirichlet-kernel term by a superposition of three Gaussian components centered at $0$ and $\pm 2\pi$, defined as
\begin{align}
&g^{\nu}_{k,i}\!=\!N_{\nu}^2\bigg(\!e^{-\tfrac{(N_{\nu}^2-1)}{12}(\psi^{\nu}_{k,i})^{2}}\!\!+\!e^{-\tfrac{(N_{\nu}^2-1)}{12}(\psi^{\nu}_{k,i}-2\pi)^{2}} \nonumber\\
&~~~~~~~~~~~~~~~~~~~~~~+\!e^{-\tfrac{(N_{\nu}^2-1)}{12}(\psi^{\nu}_{k,i}+2\pi)^{2}}\!\bigg),~ \nu \!\in\! \{x, y\},
\label{eq:g_nu}
\end{align}
and, by leveraging the separable UPA structure, we approximate the two-dimensional steering-vector correlation as
\begin{align}
\left|\mathbf{a}_k^{H}(\theta_k,\phi_k)\mathbf{a}_i(\theta_i,\phi_i)\right|^{2}
\approx g^{x}_{k,i}  g^{y}_{k,i}.
\label{eq:gauss_superposition}
\end{align}

The surrogate in \eqref{eq:gauss_superposition} is designed to capture the dominant peaks of the Dirichlet-kernel correlation term, namely, the main peak around $\psi^\nu_{k,i} \!\approx\! 0$ and the two periodic peaks near $\psi^\nu_{k,i} \!\approx\! \pm 2\pi$. However, as shown by the black solid curve in Fig.~\ref{fig2}, the exact Dirichlet-kernel profile in \eqref{dirichlet_prod} is not a smooth unimodal function; outside the main lobe, it exhibits oscillatory sidelobes with non-negligible local maxima. Since the Gaussian-mixture surrogate in \eqref{eq:g_nu} provides a smooth, non-oscillatory envelope that decays away from the dominant peaks, it may fall below these sidelobe maxima and thus underestimate the true correlation in the sidelobe region, as illustrated by the blue dash-dot curve in Fig.~\ref{fig2}. Such underestimation is undesirable for our purpose, because it leads to an underestimation of the MRT-induced inter-user interference and hence an overly optimistic interference model. To avoid this issue, we impose a conservative floor on the surrogate so that its sidelobe level is not underestimated. Specifically, we introduce an axis-wise floor constant $\epsilon_{\nu}$, chosen as the peak magnitude of the first sidelobe of the Dirichlet-kernel term associated with the $N_\nu$-element ULA, i.e., the largest sidelobe peak. It is computed numerically as
\begin{align}
\epsilon_{\nu} \triangleq \max_{\psi \in \mathcal{S}_\nu}
\left|\frac{\sin\!\left(\frac{N_\nu \psi}{2}\right)}{\sin\!\left(\frac{\psi}{2}\right)}\right|^{2},
\end{align}
where $\mathcal{S}_\nu \!=\! \left[\frac{2\pi}{N_\nu},\frac{4\pi}{N_\nu}\right]$ denotes the first-sidelobe region between the first two nulls for $\nu \!\in\! \{x,y\}$.

\begin{figure}[t]
    \centerline{\includegraphics[width=0.8\linewidth]{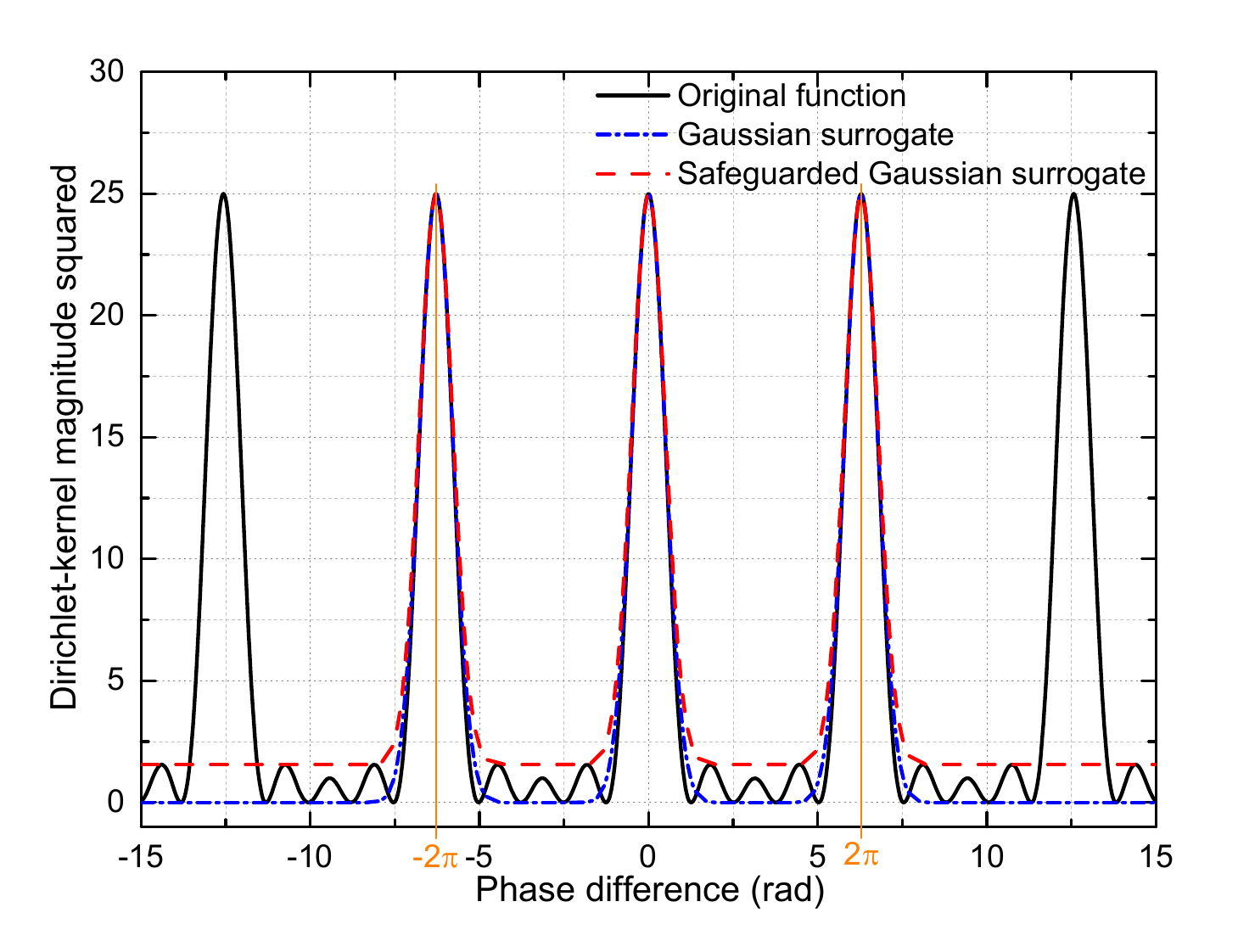}} \vspace{-4mm}
    \caption{Comparison of Dirichlet-kernel magnitude squared $\left|\sin\!\left(\frac{N_{\nu}\psi^{\nu}_{k,i}}{2}\right)/\sin\!\left(\frac{\psi^{\nu}_{k,i}}{2}\right)\right|^{2}$ and its Gaussian surrogates versus phase difference $\psi^\nu_{k,i}$ with $N_{\nu}\!=\!5$.} \vspace{-6mm}
    \label{fig2}
\end{figure}

Accordingly, we incorporate the sidelobe floor into the Gaussian surrogate through an affine scaling and shifting, so that the dominant peak structure around $\psi^\nu_{k,i} \!\approx\! 0,\pm 2\pi$ is retained while sidelobe underestimation is avoided. This yields a conservative surrogate for the sidelobe region. The resulting safeguarded surrogate for each axis, shown by the red dashed curve in Fig.~\ref{fig2}, is defined as
\begin{align}
&\hat{g}^{\nu}_{k,i}
\triangleq \big(N_{\nu}^2\!-\!\epsilon_{{\nu}}\big)\bigg(
e^{-\tfrac{(N_{\nu}^2-1)}{12}(\psi^{\nu}_{k,i})^{2}}
\!+\!e^{-\tfrac{(N_{\nu}^2-1)}{12}(\psi^{\nu}_{k,i}-2\pi)^{2}} \nonumber \\
&\hspace{2cm}
+\!e^{-\tfrac{(N_{\nu}^2-1)}{12}(\psi^{\nu}_{k,i}+2\pi)^{2}}
\bigg)\!+\!\epsilon_{{\nu}}, ~\nu \!\in\! \{x, y\}.
\label{eq:ghat_nu}
\end{align}
Finally, using \eqref{eq:ghat_nu}, the two-dimensional correlation surrogate is constructed as
\begin{align}
\hat{g}_{k,i} \triangleq \hat{g}^{x}_{k,i} \hat{g}^{y}_{k,i},
\label{eq:ghat_2d}
\end{align}
which serves as a safeguarded approximation of $\left|\mathbf{a}_k^{H}(\theta_k,\phi_k)\mathbf{a}_i(\theta_i,\phi_i)\right|^{2}$ over $\psi^\nu_{k,i}\!\in\![-2\pi,2\pi]$.

By substituting $\hat{g}_{k,i}$ into the interference term in \eqref{inter}, we obtain the following tractable SE surrogate for GN $k$:
\begin{align}
    \tilde{R}_{k} \!\triangleq\! \log_2\!\Bigg(\!1\!+\!\frac{p_k\Big(\!\frac{K_R}{K_R+1}N_xN_y\!+\!\frac{1}{K_R+1}\!\Big)}{\sum_{i \neq k}p_i\Big(\!\frac{K_R}{K_R+1}\frac{\hat{g}_{k,i}}{N_xN_y} \!+\! \frac{1}{K_R+1}\!\Big) \!+\! \sigma^2\frac{\|\mathbf{q}-\mathbf{u}_k\|^{\alpha}}{\beta}}\!\Bigg). \label{rkgk}
\end{align}
The original SE constraint in \eqref{constrk} is then approximated by 
\begin{align}
\tilde{R}_{k} \geq \eta, ~\forall k. \label{constrk10}
\end{align}

Hence, under fixed MRT beamforming based on the mean channel, the UAV positioning subproblem is formulated as
\begin{align} 
\textbf{(SP1):} ~\max_{\mathbf{q},~\eta} & ~~~~~~~~ \eta \nonumber \\ 
\textrm{subject to} & ~~~~ \eqref{qmax},~\eqref{constrk10}. \nonumber 
\end{align}
However, \textbf{(SP1)} remains nonconvex with respect to the UAV position $\mathbf{q}$. In particular, although the original steering-vector correlation term is replaced by the smooth surrogate $\hat{g}_{k,i}$, the resulting SE constraint \eqref{constrk10} still involves highly nonlinear and coupled functions of $\mathbf{q}$ through both the distance-dependent path-loss term and the geometry-dependent phase-difference terms embedded in $\hat{g}_{k,i}$. Therefore, \textbf{(SP1)} is still intractable to solve directly. To address this issue, we next develop an efficient iterative solution.

\subsubsection{Problem Solving}

To solve \textbf{(SP1)}, we construct a convex problem under the SCA framework and solve it iteratively until the UAV-position update converges. The main difficulty is that $\psi^{\nu}_{k,i}$ in $\hat{g}_{k,i}$ is a nonlinear function of the UAV position $\mathbf{q}$, since it is determined by the elevation and azimuth angles $(\theta_k,\phi_k)$ and $(\theta_i,\phi_i)$, which vary with $\mathbf{q}$, and also appears inside the exponential functions in $\hat{g}^{\nu}_{k,i}$. This coupling between the nonlinear phase-difference terms and the exponential functions, together with the product form $\hat{g}_{k,i}\!=\!\hat{g}^{x}_{k,i}\hat{g}^{y}_{k,i}$, makes the positioning problem nonconvex in $\mathbf{q}$ and prevents a direct convex reformulation. Therefore, we next develop a tractable reformulation based on convex bounding and auxiliary-variable decoupling. Since the derivations for $\nu\!=\!x$ and $\nu\!=\!y$ are identical, we present only the case $\nu\!=\!x$ without loss of generality.

Recall that the steering-vector correlation term $\left|\mathbf{a}_k^{H}(\theta_k,\phi_k)\mathbf{a}_i(\theta_i,\phi_i)\right|^{2}$ appears in the interference component. To obtain a conservative approximation of the SE constraint in \eqref{constrk10}, we construct a lower bound on the achievable SE by upper bounding the interference term in \eqref{inter}. Therefore, it is crucial to derive a tractable upper bound for the surrogate correlation $\hat g_{k,i}$. Since $\hat g^{\nu}_{k,i}$ in \eqref{eq:ghat_nu} is a positive weighted sum of exponential terms of the form $\exp\!\big(\!-b(\psi^{\nu}_{k,i})^2\big)$ with $b\!>\!0$, and $\exp(-bx)$ is monotonically decreasing in $x$, an upper bound on $\hat g^{\nu}_{k,i}$ can be obtained by constructing a tractable concave lower bound on $(\psi^{\nu}_{k,i})^2$. This is particularly useful because the resulting lower-bounding form can be further incorporated into the subsequent convex approximation procedure. In particular, to facilitate such a reformulation, we explicitly rewrite the phase-difference term $(\psi^{x}_{k,i})^{2}$ as a function of the UAV position $\mathbf{q}$. Using the identities $(x_k\!-\!x)^2\!=\!\|\mathbf{q}\!-\!\mathbf{u}_k\|^2\!-\!(y_k\!-\!y)^2\!-\!(z_k\!-\!z)^2$ and rearranging the product term involving $(x\!-\!x_k)(x\!-\!x_i)$, $(\psi^{x}_{k,i})^2$ can be expressed as
\begin{align}
&(\psi^{x}_{k,i})^2=\frac{(2\pi d)^2}{\lambda^2}\Bigg(\!\frac{x_{k}\!-\!x}{\|\mathbf{q}\!-\!\mathbf{u}_k\|}-\frac{x_{i}\!-\!x}{\|\mathbf{q}\!-\!\mathbf{u}_i\|}\!\Bigg)^2 \nonumber\\
&~~=\frac{(2\pi d)^2}{\lambda^2}\Bigg(\!2\!-\!\frac{(y_k\!-\!y)^2\!+\!(z_k\!-\!z)^2}{\|\mathbf{q}\!-\!\mathbf{u}_k\|^2} \!-\!\frac{(y_i\!-\!y)^2\!+\!(z_i\!-\!z)^2}{\|\mathbf{q}\!-\!\mathbf{u}_i\|^2}
 \nonumber\\
&~~~~~~~ \!-\!\frac{2\left(x\!-\!\frac{x_k+x_i}{2}\right)^2}{\|\mathbf{q}\!-\!\mathbf{u}_k\|\|\mathbf{q}\!-\!\mathbf{u}_i\|} \!+\! \frac{(x_k\!-\!x_i)^2}{2\|\mathbf{q}\!-\!\mathbf{u}_k\|\|\mathbf{q}\!-\!\mathbf{u}_i\|}\Bigg). \label{psi1}
\end{align}

To obtain a tractable convex approximation of \eqref{psi1}, we first decouple the distance-dependent denominators ($\|\mathbf{q}\!-\!\mathbf{u}_k\|^2$ and $\|\mathbf{q}\!-\!\mathbf{u}_i\|^2$) and the coupled product term ($\|\mathbf{q}\!-\!\mathbf{u}_k\|\|\mathbf{q}\!-\!\mathbf{u}_i\|$) by introducing nonnegative slack variables $\underline{\mu}_{k}$, $\overline{\mu}_{k}$, $\underline{\kappa}_{k,i}$, and $\overline{\kappa}_{k,i}$ as follows:
\begin{align}
&\underline{\mu}_{k} \le \|\mathbf{q}-\mathbf{u}_k\|^2, ~\forall k, \label{s} \\ 
&\overline{\mu}_{k} \ge \|\mathbf{q}-\mathbf{u}_k\|^2, ~\forall k, \label{s'} \\
&\underline{\kappa}_{k,i} \le \sqrt{\underline{\mu}_{k} \underline{\mu}_{i}}, ~\forall k,i, ~i\!\neq\! k, \label{t} \\
&\overline{\kappa}_{k,i} \ge \sqrt{\overline{\mu}_{k} \overline{\mu}_{i}}, ~\forall k,i, ~i\!\neq\! k. \label{t'}
\end{align}
Here, $\underline{\mu}_{k}$ and $\overline{\mu}_{k}$ provide lower and upper bounds on $\|\mathbf{q}\!-\!\mathbf{u}_k\|^2$, respectively, while $\underline{\kappa}_{k,i}$ and $\overline{\kappa}_{k,i}$ provide lower and upper bounds on $\|\mathbf{q}\!-\!\mathbf{u}_k\|\|\mathbf{q}\!-\!\mathbf{u}_i\|$, respectively.

Note that constraints \eqref{s} and \eqref{t'} are nonconvex: \eqref{s} imposes $\underline{\mu}_{k}$ to lie below the convex function $\|\mathbf{q}\!-\!\mathbf{u}_k\|^2$, while \eqref{t'} enforces $\overline{\kappa}_{k,i}$ to be above the concave function $\sqrt{\overline{\mu}_{k} \overline{\mu}_{i}}$. Following the SCA principle, we replace the right-hand side of \eqref{s} by its first-order Taylor lower bound and that of \eqref{t'} by its first-order Taylor upper bound at the previous iterate, which yields the following conservative convex approximations:
\begin{align}
\underline{\mu}_{k} &\le \|\mathbf{q}^{\textrm{prev}}\!-\!\mathbf{u}_k\|^2\!+\!2\big(\mathbf{q}^{\textrm{prev}}\!-\!\mathbf{u}_k\big)^{T}\big(\mathbf{q}\!-\!\mathbf{q}^{\textrm{prev}}\big), ~\forall k,
\label{s_taylor} \\
\overline{\kappa}_{k,i} &\ge \sqrt{\overline{\mu}_{k}^{\textrm{prev}}\overline{\mu}_{i}^{\textrm{prev}}}
\!+\!\frac{1}{2}\sqrt{\frac{\overline{\mu}_{i}^{\textrm{prev}}}{\overline{\mu}_{k}^{\textrm{prev}}}}\big(\overline{\mu}_{k}\!-\!\overline{\mu}_{k}^{\textrm{prev}}\big) \nonumber \\
&~~~~~+\frac{1}{2}\sqrt{\frac{\overline{\mu}_{k}^{\textrm{prev}}}{\overline{\mu}_{i}^{\textrm{prev}}}}\big(\overline{\mu}_{i}\!-\!\overline{\mu}_{i}^{\textrm{prev}}\big), ~\forall k,i, ~i\!\neq\! k,
\label{t'_taylor}
\end{align}
where $\mathbf{q}^{\textrm{prev}}$, $\overline{\mu}_{k}^{\textrm{prev}}$, and $\overline{\mu}_{i}^{\textrm{prev}}$ are the corresponding values at the previous SCA iterate.

With these slack variables, \eqref{psi1} can be lower bounded as 
\begin{align}
&(\psi^{x}_{k,i})^2\geq\frac{(2\pi d)^2}{\lambda^2}\Bigg(\!2\!-\!\frac{(y_k\!-\!y)^2\!+\!(z_k\!-\!z)^2}{\underline{\mu}_{k}}
 \nonumber\\
&-\!\frac{(y_i\!-\!y)^2\!+\!(z_i\!-\!z)^2}{\underline{\mu}_{i}} \!-\!\frac{2\left(x\!-\!\frac{x_k+x_i}{2}\right)^2}{\underline{\kappa}_{k,i}} \!+\!\frac{(x_k\!-\!x_i)^2}{2\overline{\kappa}_{k,i}}\!\Bigg).
\label{psi2}
\end{align}
Since the quadratic-over-linear function $\frac{\|\mathbf{y}\|^2}{x}$ is jointly convex for $x\!>\!0$, the first three ratio terms in \eqref{psi2} are jointly convex in the UAV position and the corresponding slack variables; hence, with the negative signs, they are jointly concave. In contrast, the last term in \eqref{psi2}, namely $\frac{(x_k-x_i)^2}{2\overline{\kappa}_{k,i}}$, is convex in $\overline{\kappa}_{k,i}\!>\!0$. Therefore, its first-order Taylor approximation at the previous iterate $\overline{\kappa}_{k,i}^{\textrm{prev}}$ provides the following lower bound:
\begin{align}
\frac{({x}_{k}\!-\!{x}_{i})^2}{2\overline{\kappa}_{k,i}} \!\geq\! -\frac{({x}_{k}\!-\!{x}_{i})^2}{2(\overline{\kappa}_{k,i}^{\textrm{prev}})^2}\big(\overline{\kappa}_{k,i}\!-\!\overline{\kappa}_{k,i}^{\textrm{prev}}\big)\!+\!\frac{({x}_{k}\!-\!{x}_{i})^2}{2\overline{\kappa}_{k,i}^{\textrm{prev}}}  \triangleq \chi_{k,i}. 
\end{align}
Substituting $\chi_{k,i}$ into the last term of \eqref{psi2}, we obtain the following concave lower-bounding surrogate for $(\psi^{x}_{k,i})^2$:
\begin{align}
&(\psi^{x}_{k,i})^2\geq\frac{(2\pi d)^2}{\lambda^2}\Bigg(\!2\!-\!\frac{(y_k\!-\!y)^2\!+\!(z_k\!-\!z)^2}{\underline{\mu}_{k}}
 \nonumber\\
&-\!\frac{(y_i\!-\!y)^2\!+\!(z_i\!-\!z)^2}{\underline{\mu}_{i}} \!-\!\frac{2\left(x\!-\!\frac{x_k+x_i}{2}\right)^2}{\underline{\kappa}_{k,i}} \!+\! \chi_{k,i}\!\Bigg) \!\triangleq \Xi^{x}_{k,i}.
\end{align}

We next derive explicit lower and upper bounds on the phase-difference term $\psi^{x}_{k,i}$ itself. This is necessary because the surrogate $\hat{g}^{x}_{k,i}$ in \eqref{eq:ghat_nu} contains the shifted terms $(\psi^{x}_{k,i} \!\pm\! 2\pi)^2$, which expand to $(\psi^{x}_{k,i})^2 \!\pm\! 4\pi\psi^{x}_{k,i} \!+\! 4\pi^2$ and therefore involve linear terms in $\psi^{x}_{k,i}$. Consequently, bounding the shifted components requires explicit bounds on $\psi^{x}_{k,i}$ (in addition to bounds on $(\psi^{x}_{k,i})^2$). To this end, we first rewrite $\psi^{x}_{k,i}$ as a function of the UAV position $\mathbf{q}$, yielding
\begin{align}
\psi^{x}_{k,i}\!=\!\frac{2\pi d}{\lambda}\Bigg(\!\frac{{x}_{k}}{\|\mathbf{q}\!-\!\mathbf{u}_k\|}\!-\!\frac{{x}}{\|\mathbf{q}\!-\!\mathbf{u}_k\|}\!-\!\frac{{x}_{i}}{\|\mathbf{q}\!-\!\mathbf{u}_i\|}\!+\!\frac{{x}}{\|\mathbf{q}\!-\!\mathbf{u}_i\|}\!\Bigg). \label{psix_expand}
\end{align}
According to constraint \eqref{qmax}, the UAV and all GNs lie in the nonnegative region, so that $x,x_k,x_i \!\ge\! 0$ over the feasible set. This property is important for constructing valid bounds for the ratio terms in \eqref{psix_expand}, because the monotonicity of a ratio $\frac{a}{y}$ with respect to the positive denominator $y\!>\!0$ depends on the sign of the numerator $a$. Specifically, for any fixed $a \!\ge\! 0$, $\frac{a}{y}$ is monotonically decreasing in $y$, whereas for $a \!\le\! 0$, it is monotonically increasing. Therefore, if the numerators in \eqref{psix_expand} are allowed to change sign, the inequality directions would become sign-dependent, and a single consistent bounding construction would not be available. Accordingly, under the nonnegative-domain condition $x,x_k,x_i \!\ge\! 0$, the subsequent bounds can be derived consistently.

From the expansion $(\psi^{x}_{k,i}\pm 2\pi)^2\!=\!(\psi^{x}_{k,i})^2 \!\pm\! 4\pi \psi^{x}_{k,i} \!+\! 4\pi^2$, it follows that the $+2\pi$-shifted term requires a concave lower bound on $\psi^{x}_{k,i}$, whereas the $-2\pi$-shifted term requires a convex upper bound on $\psi^{x}_{k,i}$.

\noindent\textit{i) $+2\pi$-shifted term:}
From \eqref{psix_expand}, a lower bound on $\psi^{x}_{k,i}$ is obtained by lower bounding the positive terms and upper bounding the negative terms. Specifically,
\begin{align}
\psi^{x}_{k,i}
&\geq \frac{2\pi d}{\lambda}\Bigg(\frac{x_{k}}{\overline{\xi}_k}-\frac{{x}}{\underline{\xi}_k}-\frac{{x}_{i}}{\underline{\xi}_i}+\frac{{x}}{\overline{\xi}_i}\Bigg), \label{psi_plus_step1}
\end{align}
where the slack variables $\underline{\xi}_k$ and $\overline{\xi}_k$ are introduced to decouple the distance-dependent denominators as
\begin{align}
&\underline{\xi}_k \le \|\mathbf q-\mathbf u_k\|, ~\forall k, \label{xi}\\
&\overline{\xi}_k \ge \|\mathbf q-\mathbf u_k\|, ~\forall k. \label{xi'}
\end{align}
Since \eqref{xi} is nonconvex, we replace its right-hand side by the first-order lower bound at $\mathbf q^{\textrm{prev}}$, which yields
\begin{align}
\underline{\xi}_k \leq \|\mathbf{q}^{\textrm{prev}}\!-\!\mathbf{u}_k\|\!+\!\frac{(\mathbf{q}^{\textrm{prev}}\!-\!\mathbf{u}_k)^T}{\|\mathbf{q}^{\textrm{prev}}\!-\!\mathbf{u}_k\|}\big(\mathbf{q}\!-\!\mathbf{q}^{\textrm{prev}}\big), ~\forall k. \label{xi_taylor} 
\end{align}

Moreover, since $\frac{x_{k}}{\overline{\xi}_k}$ is convex in $\overline{\xi}_k \!>\! 0$, its first-order Taylor approximation at $\overline{\xi}^{\textrm{prev}}_k$ provides the following lower bound:
\begin{align}
\frac{x_{k}}{\overline{\xi}_k} \ge \frac{{x}_{k}}{\overline{\xi}_k^{\textrm{prev}}}-\frac{x_{k}}{(\overline{\xi}_k^{\textrm{prev}})^2}\big(\overline{\xi}_k-\overline{\xi}_k^{\textrm{prev}}\big). \label{xkxi}
\end{align}

To construct a tractable lower bound for the negative ratio term $-\frac{x}{\underline{\xi}_k}$, we introduce an auxiliary variable $\rho$ such that
\begin{align}
x \le \rho^2. \label{rho}
\end{align}
Since $\underline{\xi}_k \!>\! 0$, this immediately implies $-\frac{x}{\underline{\xi}_k} \!\ge\! -\frac{\rho^2}{\underline{\xi}_k}$. The resulting term has a negative quadratic-over-linear form, which is jointly concave in $(\rho,\underline{\xi}_k)$ and will be useful in the subsequent convexification. However, the auxiliary constraint \eqref{rho} is nonconvex. To obtain a tractable reformulation, we replace $\rho^2$ by its first-order lower bound at $\rho^{\textrm{prev}}$, yielding
\begin{align}
x \le (\rho^{\textrm{prev}})^2+2\rho^{\textrm{prev}}(\rho-\rho^{\textrm{prev}}). \label{rho_taylor}
\end{align}

Applying the quadratic transform \cite{Shen18} to the positive ratio term $\frac{x}{\overline{\xi}_i}$ yields the lower bound with an auxiliary variable $\tau_i$:
\begin{align}
\frac{x}{\overline{\xi}_i}\ge 2\tau_i\sqrt{x}-\tau_i^2\overline{\xi}_i. \label{QT}
\end{align}

Finally, by substituting \eqref{xkxi} for $\frac{x_{k}}{\overline{\xi}_k}$, $-\frac{\rho^2}{\underline{\xi}_k}$ for $-\frac{{x}}{\underline{\xi}_k}$, and \eqref{QT} for $\frac{{x}}{\overline{\xi}_i}$, while retaining the term $-\frac{{x}_{i}}{\underline{\xi}_i}$, \eqref{psi_plus_step1} can be further lower bounded by the following tractable concave surrogate:
\begin{align}
\psi^{x}_{k,i}
&\geq \frac{2\pi d}{\lambda}\Bigg(\frac{{x}_{k}}{\overline{\xi}_k^{\textrm{prev}}}-\frac{x_{k}}{(\overline{\xi}_k^{\textrm{prev}})^2}\big(\overline{\xi}_k-\overline{\xi}_k^{\textrm{prev}}\big)
-\frac{\rho^2}{\underline{\xi}_k}-\frac{{x}_{i}}{\underline{\xi}_i} \nonumber \\
&~~~~~~~~~~~~~+2\tau_{i}\sqrt{{x}}-\tau_{i}^2{\overline{\xi}_i}\Bigg)
\triangleq \underline{\psi}^{x}_{k,i}. \label{psi_plus}
\end{align}

\noindent\textit{ii) $-2\pi$-shifted term:}
Similarly, to handle the $-2\pi$-shifted term, we construct an upper bound on $\psi^{x}_{k,i}$ as
\begin{align}
\psi^{x}_{k,i}
&\leq \frac{2\pi d}{\lambda}\Bigg(\frac{x_{k}}{\underline{\xi}_{k}}-\frac{{x}}{\overline{\xi}_{k}}-\frac{x_{i}}{\overline{\xi}_{i}}+\frac{{x}}{\underline{\xi}_{i}}\Bigg) \nonumber \\
&\leq \frac{2\pi d}{\lambda}\Bigg(\frac{{x}_{k}}{\underline{\xi}_{k}}-\big(2\tau_{k}\sqrt{x}-\tau_{k}^2\overline{\xi}_{k}\big)
-\frac{{x}_{i}}{\overline{\xi}_{i}^{\textrm{prev}}}\nonumber \\
&~~~~~~~~~~~~~+\frac{x_{i}}{(\overline{\xi}_{i}^{\textrm{prev}})^2}\big(\overline{\xi}_{i}-\overline{\xi}_{i}^{\textrm{prev}}\big)+\frac{\rho^2}{\underline{\xi}_{i}}\Bigg)
\triangleq \overline{\psi}^{x}_{k,i}. \label{psi_minus}
\end{align}
The second inequality is obtained by applying the same bounding techniques as in the $+2\pi$-shifted case, including the quadratic-transform bound, the first-order Taylor approximation, and the $\rho$-based ratio surrogate. However, because the $-2\pi$ shift changes the sign of the linear term to $-4\pi\psi^{x}_{k,i}$ in $(\psi^{x}_{k,i}-2\pi)^2$, the bound directions are reversed: positive terms in $\psi^{x}_{k,i}$ are upper bounded, while negative terms are lower bounded. As a result, $\overline{\psi}^{x}_{k,i}$ serves as a tractable convex upper-bounding surrogate for $\psi^{x}_{k,i}$.

For both the $+2\pi$- and $-2\pi$-shifted terms, the quadratic-transform term has the form $2\tau\sqrt{x}\!-\!\tau^2\overline{\xi}$, which is concave in $\tau$ for $\overline{\xi}\!>\!0$. Therefore, the tightest bound is obtained by maximizing it with respect to the corresponding auxiliary variable. This yields the closed-form updates
\begin{align}
\tau_i^{*}=\frac{\sqrt{x}}{\overline{\xi}_{i}},~\forall i \neq k,
\qquad
\tau_k^{*}=\frac{\sqrt{x}}{\overline{\xi}_{k}},~\forall k, \label{tau_opt}
\end{align}
where $\pmb{\tau}\!\triangleq\!\{\tau_k,\forall k\}$ denotes the set of auxiliary variables.

By substituting $\Xi^{x}_{k,i}$, $\underline{\psi}^{x}_{k,i}$, and $\overline{\psi}^{x}_{k,i}$ into \eqref{eq:ghat_nu}, we obtain the following upper bound for $\hat{g}^{x}_{k,i}$:
\begin{align}
\!\!\!&\hat{g}^{x}_{k,i}
\!\le\! (N_{x}^2\!-\!\epsilon_{{x}})\!\bigg(\!
e^{-\tfrac{(N_{x}^2-1)}{12}\Xi^{x}_{k,i}}
\!+\!e^{-\tfrac{(N_{x}^2-1)}{12}(\Xi^{x}_{k,i} + 4\pi \underline{\psi}^{x}_{k,i} + 4\pi^2)} \nonumber \\
\!\!\!&~~~~~~~~~~~+\!e^{-\tfrac{(N_{x}^2-1)}{12}(\Xi^{x}_{k,i} - 4\pi \overline{\psi}^{x}_{k,i} + 4\pi^2)}\!
\bigg)\!+\!\epsilon_{{x}} \triangleq \hat{g}^{x,\textrm{UB}}_{k,i}.
\label{eq:ghat_nu2}
\end{align}
Since $\Xi^{x}_{k,i}$ and $\underline{\psi}^{x}_{k,i}$ are concave, whereas $\overline{\psi}^{x}_{k,i}$ is convex, all three exponential arguments in \eqref{eq:ghat_nu2} are convex. Therefore, $\hat{g}^{x,\textrm{UB}}_{k,i}$ is a convex upper-bounding surrogate of $\hat{g}^{x}_{k,i}$.

Following the same procedure, we can similarly construct the upper-bounding surrogate $\hat{g}^{y,\textrm{UB}}_{k,i}$. Hence, the two-dimensional correlation term $\left|\mathbf{a}_k^{H}(\theta_k,\phi_k)\mathbf{a}_i(\theta_i,\phi_i)\right|^{2}$ admits the following upper bound:
\begin{align}
\left|\mathbf{a}_k^{H}(\theta_k,\phi_k)\mathbf{a}_i(\theta_i,\phi_i)\right|^{2}
\le \hat{g}^{x,\textrm{UB}}_{k,i}\hat{g}^{y,\textrm{UB}}_{k,i}. 
\end{align}
To avoid the direct use of this product in the SE constraint, we introduce a slack variable $\Psi_{k,i}$ and impose
\begin{align}
\hat{g}^{x,\textrm{UB}}_{k,i}\hat{g}^{y,\textrm{UB}}_{k,i} \le \Psi_{k,i}, ~\forall k,i,\; i\neq k.
\label{eq:ghat_2d2}
\end{align}
Here, both $\hat{g}^{x,\textrm{UB}}_{k,i}$ and $\hat{g}^{y,\textrm{UB}}_{k,i}$ are convex, since each is a finite sum of exponential terms with convex arguments. Although the product of two convex functions is not convex in general, $\hat{g}^{x,\textrm{UB}}_{k,i}\hat{g}^{y,\textrm{UB}}_{k,i}$ can be rewritten as a finite sum of exponential terms with convex arguments and nonnegative coefficients, and is therefore convex. Hence, the left-hand side of \eqref{eq:ghat_2d2} is convex, and \eqref{eq:ghat_2d2} defines a convex constraint that can be efficiently handled within the proposed optimization framework.

By applying the upper-bounding slack variable $\Psi_{k,i}$ to the inter-user correlation term, the achievable SE in \eqref{rkgk} can be lower bounded as
\begin{align}
    \tilde{R}_{k} &\geq \log_2\!\Bigg(\!1\!+\!\frac{p_k\Big(\frac{K_R}{K_R+1}N_xN_y\!+\!\frac{1}{K_R+1}\Big)}{\sum_{i \neq k}p_i\Big(\!\frac{K_R}{K_R+1}\frac{\Psi_{k,i}}{N_xN_y} \!+\! \frac{1}{K_R+1}\!\Big) \!+\!\sigma^2\frac{\|\mathbf{q}-\mathbf{u}_k\|^{\alpha}}{\beta}}\!\Bigg) \nonumber \\
    &\geq \log_2\!\Bigg(\!1\!+\!2\upsilon_k\sqrt{p_k\Big(\frac{K_R}{K_R\!+\!1}N_xN_y\!+\!\frac{1}{K_R\!+\!1}\Big)} \nonumber \\ 
    &-\upsilon_k^2\bigg(\sum_{i \neq k}p_i\Big(\frac{K_R}{K_R\!+\!1}\frac{\Psi_{k,i}}{N_xN_y} \!+\! \frac{1}{K_R\!+\!1}\Big) \!+\! \sigma^2\frac{\|\mathbf{q}\!-\!\mathbf{u}_k\|^{\alpha}}{\beta}\!\bigg)\!\!\Bigg) \nonumber \\ 
    &\triangleq \underline{R}_k,
\end{align}
where $\pmb{\upsilon}\!\triangleq\!\{\upsilon_k,\forall k\}$ denotes the set of auxiliary variables introduced by the quadratic transform \cite{Shen18}. Since $\log_2(1+z)$ is monotonically increasing in $z$, a lower bound on the fractional SINR term directly induces a lower bound on the achievable SE. Therefore, by applying the quadratic transform to the SINR term, we obtain the tractable lower bound $\underline{R}_k$. Moreover, for any fixed $\mathbf{q}$ and $\{\Psi_{k,i}\}$, this bound is tight when
\begin{align}
\!\!\!\!\upsilon_k^*\!=\!
\frac{\sqrt{p_k\Big(\!\frac{K_R}{K_R+1}N_xN_y\!+\!\frac{1}{K_R+1}\!\Big)}}
{\sum_{i \neq k}p_i\Big(\!\frac{K_R}{K_R+1}\frac{\Psi_{k,i}}{N_xN_y}\!+\! \frac{1}{K_R+1}\!\Big) \!+\!\sigma^2\frac{\|\mathbf{q}-\mathbf{u}_k\|^{\alpha}}{\beta}}, ~\forall k.
\label{vk_opt}
\end{align}
Consequently, constraint \eqref{constrk10} is conservatively replaced by
\begin{align}
\underline{R}_k \geq \eta, ~\forall k. \label{rate_fp}   
\end{align}

Accordingly, \textbf{(SP1)} is successively approximated by the following convex subproblem:
\begin{align} 
\textbf{(SP1-1):} ~\max_{\mathbf{q},~\pmb{\Lambda},~\eta} & ~~~~~~~~~~~~~~~~~~ \eta \nonumber  \\
\textrm{subject to} & ~~~~\eqref{qmax},~\eqref{s'},~\eqref{t},~\eqref{s_taylor},~\eqref{t'_taylor}, \nonumber \\
&~~~~\eqref{xi'},~\eqref{xi_taylor},~\eqref{rho_taylor},~\eqref{eq:ghat_2d2},~\eqref{rate_fp}, \nonumber
\end{align}
where $\pmb{\Lambda} \!\triangleq\! \{\underline{\mu}_{k}, \overline{\mu}_{k},
\underline{\xi}_{k}, \overline{\xi}_{k}, \forall k;
\underline{\kappa}_{k,i}, \overline{\kappa}_{k,i}, \Psi_{k,i}, \forall k,i, i \!\neq\! k;
\rho\}$ denotes the set of slack variables.

\vspace{-2mm}
\subsection{Beamforming Optimization With Instantaneous CSI}

In the second stage, the UAV position obtained from the first stage is fixed, and the beamforming design is optimized based on the instantaneous CSI $\{\mathbf{h}_k,\forall k\}$ fed back from the GNs. The resulting beamforming subproblem is given by
\begin{align} 
\textbf{(SP2):} ~\max_{\mathbf{W},~\eta} & ~~~~~~~~ \eta \nonumber  \\
\textrm{subject to} & ~~~~ \eqref{constrk},~\eqref{pmax}. \nonumber
\end{align}
Problem \textbf{(SP2)} is still nonconvex because the SE constraints in \eqref{constrk} contain coupled SINR expressions, where each beamformer affects not only the desired signal power but also the multiuser interference. This coupling makes the direct optimization of $\mathbf{W}$ difficult. To obtain a more tractable formulation, we adopt the weighted minimum mean square error (WMMSE) approach \cite{Shi11}, which reformulates the SE maximization problem into an equivalent MSE-based problem by introducing auxiliary equalizer and weight variables. The key idea is to exploit the SE--MSE equivalence, so that the difficult SINR-coupled SE expressions can be handled through an alternative optimization framework. Although the reformulated problem is still not jointly convex, it becomes much more amenable to block coordinate optimization.

Specifically, GN $k$ employs a scalar equalizer $e_k$ to detect its intended stream $s_k$ from the received signal $y_k$, i.e.,
\begin{align}
\hat{s}_k=(e_k)^*y_k.
\end{align}
Accordingly, the MSE for decoding stream $s_k$ is defined as
\begin{align}
    \varepsilon_k &\triangleq \mathbb{E}\{|\hat{s}_k-s_k|^2\} \nonumber \\ 
    &=|(e_k)^*|^2T_k-2\Re\!\left((e_k)^*\mathbf{h}^H_k\mathbf{w}_k\right)+1,
    \label{mse1}
\end{align}
where $T_k\!=\!\sum_{i=1}^K|\mathbf{h}^H_k\mathbf{w}_i|^2\!+\!\sigma^2$.

For a given beamforming matrix $\mathbf{W}$, the MSE in \eqref{mse1} is minimized by the MMSE equalizer. By setting $\frac{\partial \varepsilon_k}{\partial e_k}\!=\!0$, the optimal equalizer is obtained as
\begin{align}
    (e_k)^{\textrm{MMSE}}=\mathbf{h}^H_k\mathbf{w}_k(T_k)^{-1}.
    \label{mmse1} 
\end{align}
Substituting \eqref{mmse1} into \eqref{mse1} yields the resulting MMSE:
\begin{align}
    (\varepsilon_k)^{\textrm{MMSE}}
    = 1-\frac{|\mathbf{h}^H_k\mathbf{w}_k|^2}{T_k}
    =\frac{1}{1+\gamma_k},
\end{align}
where $\gamma_k$ denotes the SINR of GN $k$. Hence, the SINR can be equivalently written as
\begin{align}
    \gamma_k= \frac{1}{(\varepsilon_k)^{\textrm{MMSE}}}-1,
\end{align}
and the achievable SE is expressed as
\begin{align}
    R_k= -\log_2\!\left((\varepsilon_k)^{\textrm{MMSE}}\right).
\end{align}
This relation is the key step of the WMMSE method: maximizing the achievable SE is equivalent to minimizing the corresponding MMSE. In other words, instead of directly optimizing the SE expression, we may optimize an equivalent error metric.

To further obtain a form suitable for iterative optimization, we introduce an auxiliary positive weight variable $u_k$ for each GN and define the weighted MSE as
\begin{align}
    \zeta_k = u_k\varepsilon_k-\ln(u_k).
\end{align}
Together with the receive equalizer $e_k$, the auxiliary variable $u_k$ establishes an affine equivalence between the minimum weighted MSE and the achievable SE. In particular, for any fixed $\mathbf{W}$, minimizing $\zeta_k$ over $(u_k,e_k)$ yields
\begin{align}
    (\zeta_k)^{\textrm{MMSE}} = \min_{u_k,e_k} \zeta_k = 1- R_k \ln 2,
\end{align}
where the optimal weight is obtained from $\frac{\partial \zeta_k}{\partial u_k}\!=\!0$ as
\begin{align}
    u_k^*=\frac{1}{(\varepsilon_k)^{\textrm{MMSE}}}. \label{uopt}
\end{align}
Hence, once the equalizer and the weight are optimally chosen, the weighted-MSE quantity is affinely equivalent to the achievable SE. As a result, minimizing the weighted MSE is fully equivalent to maximizing the SE.

Using this equivalence, \textbf{(SP2)} can be reformulated as the following WMMSE problem:
\begin{align} 
\textbf{(SP2-1):} ~\min_{\mathbf{W},~\mathbf{u},~\mathbf{e},~\eta^{\textrm{MMSE}}} & ~~~~~~~~ \eta^{\textrm{MMSE}} \nonumber  \\
\textrm{subject to} ~~~& ~~~~ \frac{(\zeta_k)^{\textrm{MMSE}}\!-\!1}{\ln2} \!\le\! \eta^{\textrm{MMSE}},~\forall k, \\
& ~~~~~~~ \eqref{pmax}, \nonumber
\end{align}
where $\mathbf{u}\!\triangleq\!\{{u}_k, \forall k\}$ and $\mathbf{e}\!\triangleq\!\{{e}_k, \forall k\}$. This reformulation is equivalent to the original max--min SE beamforming problem. Specifically, since $\frac{(\zeta_k)^{\textrm{MMSE}}\!-\!1}{\ln2}\!=\!-R_k$, minimizing the maximum weighted-MSE expression across GNs is equivalent to maximizing the minimum achievable SE among all GNs. Although \textbf{(SP2-1)} is still not jointly convex in all variables, it is convex with respect to each block of variables when the others are fixed. Therefore, \textbf{(SP2-1)} can be efficiently solved by alternately updating $\mathbf{e}$ and $\mathbf{u}$ in closed form and optimizing $\mathbf{W}$ via convex solvers until convergence.

\vspace{-2mm}
\subsection{Proposed Algorithm}

The proposed two-stage procedure is summarized in Algorithm~\ref{Alg1}. In the first stage, the UAV position is iteratively refined via SCA: the auxiliary variables for the quadratic transform are first updated in closed form, and the resulting convex subproblem \textbf{(SP1-1)} is then solved by a standard convex solver such as CVX \cite{CVX}. In the second stage, with the UAV position fixed and the instantaneous CSI acquired, the beamforming design is optimized using the WMMSE approach: the equalizers and weights are updated in closed form, and the beamforming matrix is then obtained by solving \textbf{(SP2-1)}. These updates are repeated until convergence.

\begin{algorithm}[h]
    \caption{Proposed Algorithm} \label{Alg1} \small
    1:$~~$Initialize $\mathbf{q}^{(0)}$ and $\epsilon > 0$\\
    2:$~~$Initialize $\mathbf{W}^{(0)}$ using MRT based on the mean channel \\
    3:$~~$Set $r\!=\!0$ and calculate $\eta^{(r)} \!=\! \min_{k \in \mathcal{K}} \underline{R}_{k}$\\
    4:$~~$\textbf{repeat}\\
    5:$~~~~$Update $r \leftarrow r+1$  \\
    6:$~~~~$Update $\eta^{\textrm{old}} \leftarrow \eta^{(r-1)}$\\
    7:$~~~~$Update $\{\pmb{\tau}^{(r)},\pmb{\upsilon}^{(r)}\}$ using \eqref{tau_opt} and \eqref{vk_opt} \\
    8:$~~~~$Find $\mathbf{q}^{(r)}$ from \textbf{(SP1-1)} for $\{\mathbf{q}^{(r-1)},\pmb{\tau}^{(r)},\pmb{\upsilon}^{(r)}\}$\\
    9:$~~~~$Calculate $\eta^{(r)} \!=\! \min_{k \in \mathcal{K}} \underline{R}_{k}$\\
    10:$~$\textbf{until} $|\eta^{(r)}\!-\!\eta^{\textrm{old}}|<\epsilon$ \\
    11:$~$Set $\mathbf{q}^{*} \leftarrow \mathbf{q}^{(r)}$ and obtain the instantaneous CSI via feedback \\
    12:$~$Set $r\!=\!0$ and calculate $\eta^{(r)} \!=\! \min_{k \in \mathcal{K}} R_{k}$ \\
    13:$~$\textbf{repeat}\\
    14:$~~~$Update $r \leftarrow r+1$  \\
    15:$~~~$Update $\eta^{\textrm{old}} \leftarrow \eta^{(r-1)}$\\
    16:$~~~$Update $\{\mathbf{e}^{(r)},\mathbf{u}^{(r)}\}$ using \eqref{mmse1} and \eqref{uopt} \\
    17:$~~~$Find $\mathbf{W}^{(r)}$ from \textbf{(SP2-1)} for $\{\mathbf{q}^{*},\mathbf{e}^{(r)},\mathbf{u}^{(r)},\mathbf{W}^{(r-1)}\}$ \\
    18:$~~~$Calculate $\eta^{(r)} \!=\! \min_{k \in \mathcal{K}} R_{k}$\\
    19:$~$\textbf{until} $|\eta^{(r)}\!-\!\eta^{\textrm{old}}|<\epsilon$ 
\end{algorithm}

\begin{remark}[Convergence Analysis] \label{Remark1}
Algorithm~\ref{Alg1} consists of two sequential stages. In Stage 1, each SCA iteration solves the convex approximation \textbf{(SP1-1)}, where the nonconvex terms are replaced by conservative surrogates that are locally tight at the expansion point. As a result, the objective sequence generated by the SCA iterates in the UAV-positioning stage is monotonically non-decreasing. Moreover, the feasible set is bounded by the UAV-position constraint \eqref{qmax}, and the achievable objective value is upper-bounded over this feasible region by a finite constant \cite{Bertsekas99}. Therefore, the Stage-1 sequence converges.

In Stage 2, for the fixed UAV position, the WMMSE-based iterative updates generate a monotonically non-increasing weighted-MSE objective sequence. With MMSE-optimal updates of the equalizer and weight variables, this weighted-MSE reformulation is equivalent to the original max--min SE beamforming problem. Therefore, the corresponding minimum achievable SE sequence is monotonically non-decreasing. Since the achievable SE is upper-bounded under the transmit-power constraint \eqref{pmax}, the Stage-2 sequence also converges. Therefore, both the SCA-based UAV-positioning stage and the WMMSE-based beamforming stage generate convergent objective sequences.
\end{remark}

\begin{remark}[Computational Complexity] \label{Remark2}
The computational complexity of the proposed algorithm is evaluated using the worst-case complexity of interior-point methods~\cite{Ben-Tal01}. For an objective accuracy of $\epsilon$, such methods generally require $\mathcal{O}\!\left(\sqrt{N_V}\log(1/\epsilon)\right)$ iterations, where each iteration incurs a complexity of $\mathcal{O}(N_V^3)$ and $N_V$ is the number of optimization variables. Therefore, the overall complexity of the proposed algorithm is on the order of $\emph{O}\Big(\!\big(I_1(K^2)^{3.5}\!+\!I_2(KN_xN_y)^{3.5}\big) \log(1/\epsilon)\!\Big)$, where $I_1$ and $I_2$ are the numbers of iterations until convergence in Stages 1 and 2, respectively. Hence, the proposed method has polynomial-time complexity with respect to the problem size and is computationally tractable for practical implementation~\cite{Leiserson}.
\end{remark}

\vspace{-2mm}
\section{Simulation Results and Discussions}

\begin{table}[ht]\vspace{-4mm}
\begin{center}
\caption{Parameter Setup} \footnotesize
\begin{tabular}{ll} \hline 
Description & Value \\ \hline \hline
Number of GNs & $K$ = 6 \\
Number of antennas & $N_x$ = $N_y$ = 4 \\
Operating region & $\{x_{\textrm{max}},x_{\textrm{min}}\}$ = $\{400,100\}$ m \\
& $\{y_{\textrm{max}},y_{\textrm{min}}\}$ = $\{400,100\}$ m \\
& $\{z_{\textrm{max}},z_{\textrm{min}}\}$ = $\{100,30\}$ m \\
Maximum transmit power & $P_{\textrm{max}}$ =  30 dBm\\
Carrier frequency & $f_c$ = $3.5$ GHz \\
Antenna spacing & $d$ = $\frac{\lambda}{2}$ = 4.29 cm \\
Path-loss exponent & $\alpha$ = $2.3$ \\
Reference channel gains & $\beta$ = $-43$ dB \\
Rician $K$-factor & $K_{R}$ = 10 dB\\
Noise power & $\sigma^{2}$ = $-100$ dBm \\
Convergence threshold & $\epsilon$ = 10$^{-3}$ \\
\hline
\end{tabular} 
\label{table1}
\end{center}\vspace{-4mm}
\end{table}

For performance evaluation, we use the simulation parameters summarized in Table~\ref{table1}, which follow representative settings commonly adopted in prior UAV communication studies~\cite{Liu23,Chou25}. The GNs are randomly distributed within the $x$--$y$ region specified in Table~\ref{table1} on the horizontal plane ($z_k\!=\!0$), while the UAV is constrained to move within the corresponding feasible 3D region. For each channel realization, the air-to-ground channel is generated according to the Rician fading model described in \eqref{cch}. Based on this setup, we compare the proposed scheme with the following benchmark schemes in order to evaluate the performance gain brought by correlation-aware UAV positioning and the subsequent beamforming refinement.
\begin{enumerate}
    \item \emph{Taylor-Linearized Alternating Positioning and Beamforming (TLAPB) \cite{Xu25}}: UAV positioning and beamforming are iteratively optimized based on mean channel, where the UAV position is refined via local first-order Taylor-based SCA and the beamforming variables are updated alternately until convergence. After convergence, the instantaneous CSI at the obtained UAV location is fed back, and the transmit beamformer is re-optimized using the updated channel information.

    \item \emph{Interference-Unaware Positioning (IUP) \cite{Zhu22,Liu23,Yu22,Chou25}}: The UAV position is optimized based on long-term channel statistics for desired-signal enhancement under an idealized decoupled beam model that ignores inter-user interference. After the UAV position is obtained, the transmit beamformer is optimized using the instantaneous CSI at the resulting UAV location.

    \item \emph{Correlation-Frozen Positioning (CFP)}: The UAV position is optimized based on long-term channel statistics with the inter-user steering-correlation term $\left|\mathbf{a}_k^{H}\mathbf{a}_i\right|^{2}$ fixed at the value computed from the initial UAV location. After the UAV position is obtained, the transmit beamformer is optimized using the instantaneous CSI at the resulting UAV location.

    \item \emph{Fixed-Center Positioning (FCP)}: The UAV position is fixed at the center of the operating region, and the transmit beamformer is optimized using the instantaneous CSI.
\end{enumerate}
For fair comparison, the initial UAV position for all iterative schemes except FCP, including TLAPB, IUP, CFP, and the proposed scheme, is set to the centroid of the GN distribution. TLAPB and IUP serve as representative baselines inspired by prior studies, enabling comparison with existing UAV positioning frameworks that do not explicitly optimize the UAV location based on the proposed correlation-aware interference structure. By contrast, CFP and FCP are restricted variants of the proposed framework, introduced to isolate the gains due to dynamic correlation-aware positioning updates and to UAV positioning optimization itself, respectively.

\begin{figure*}[ht!]
  \begin{center}
    \subfigure[3D visualization of optimized UAV positions under different schemes.]{
      \includegraphics[width=0.29\linewidth]{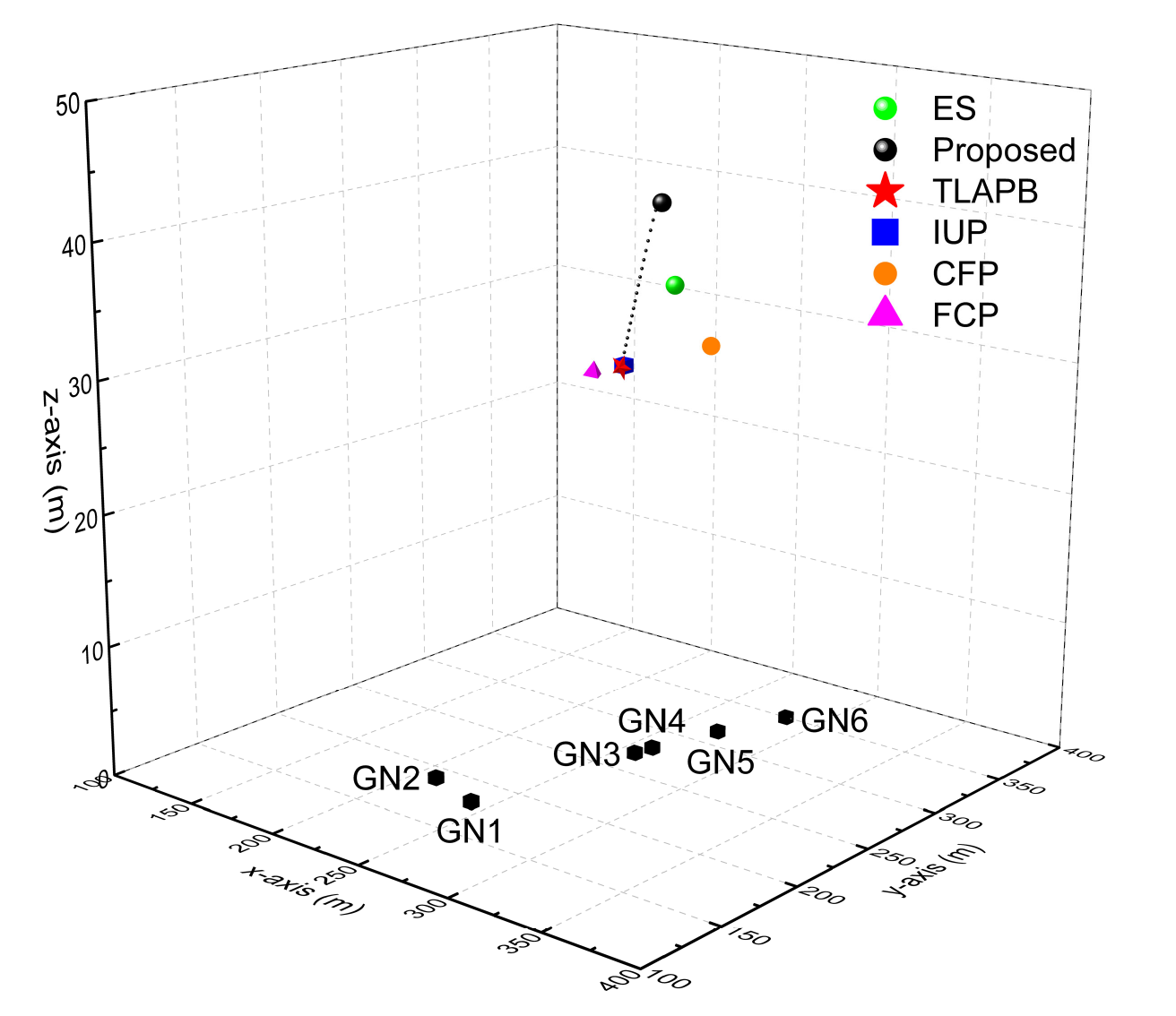}\label{R1-1}
    }
    \subfigure[2D visualization of optimized UAV positions on the achievable minimum SE heatmap.]{
      \includegraphics[width=0.315\linewidth]{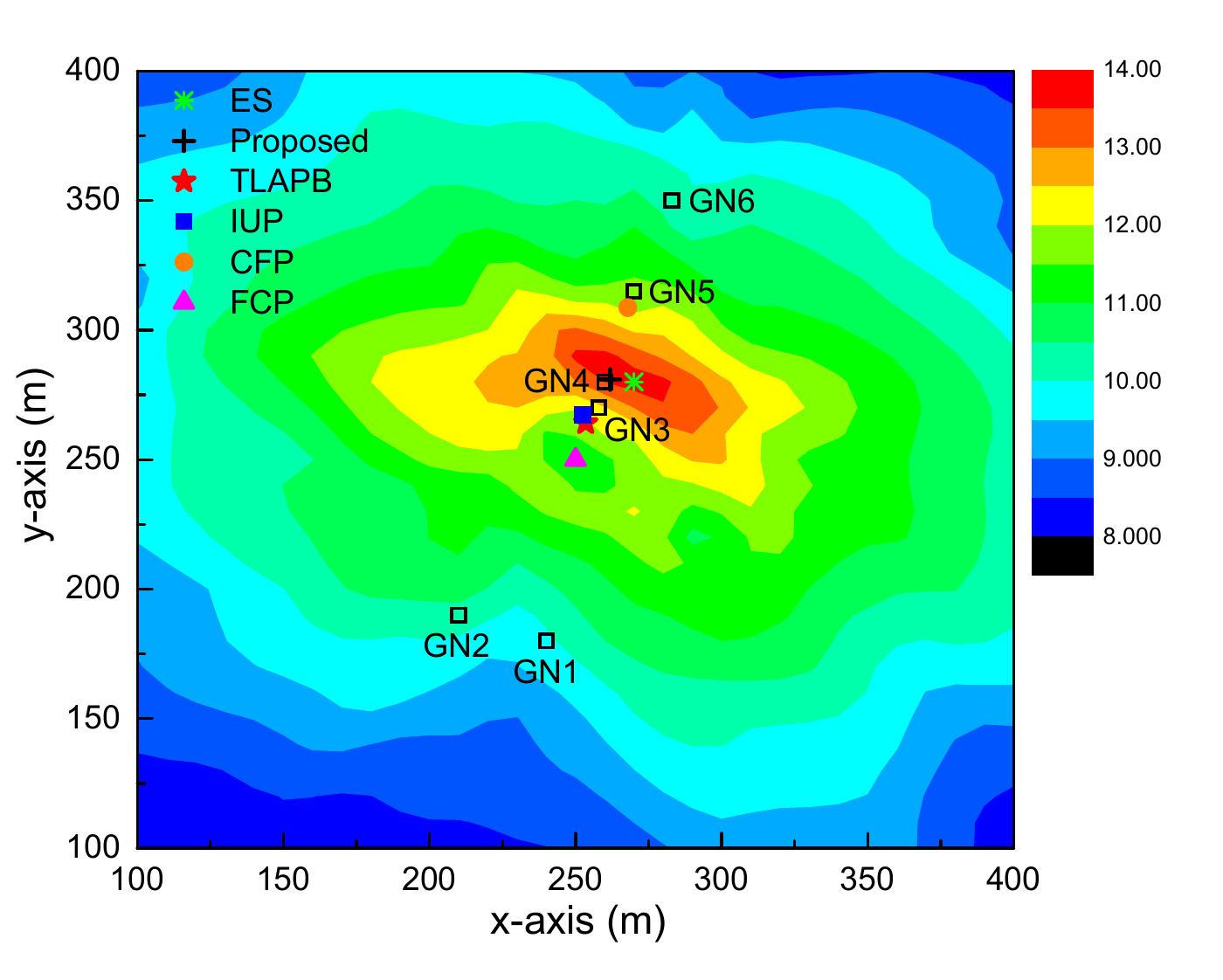}\label{R1-2}
    }
    \subfigure[Convergence behavior and maximum spatial correlation of the proposed scheme.]{
      \includegraphics[width=0.315\linewidth]{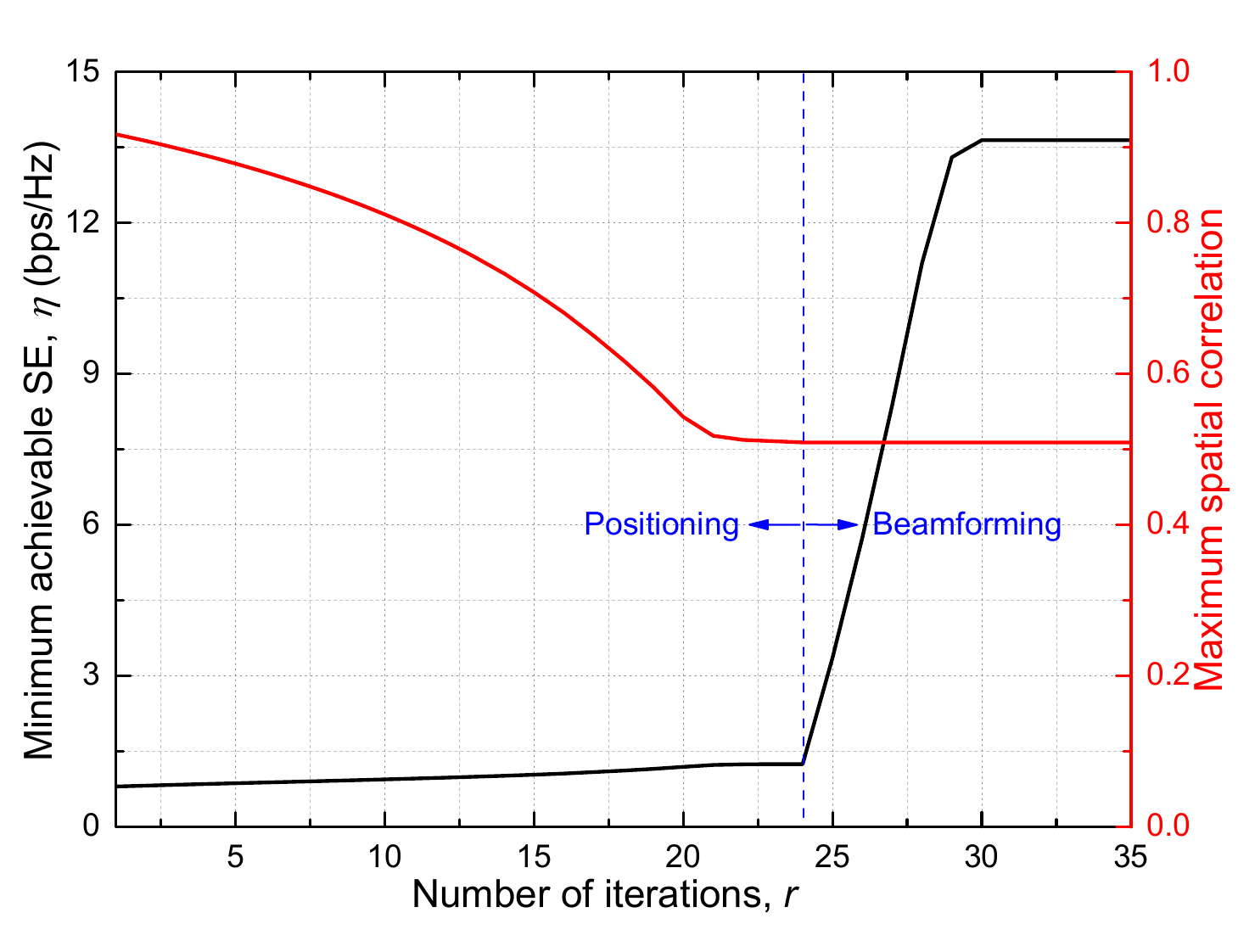}\label{R1-3}
    }
    \subfigure[Beam gain pattern for GN 5 at the initial UAV position.]{
      \includegraphics[width=0.315\linewidth]{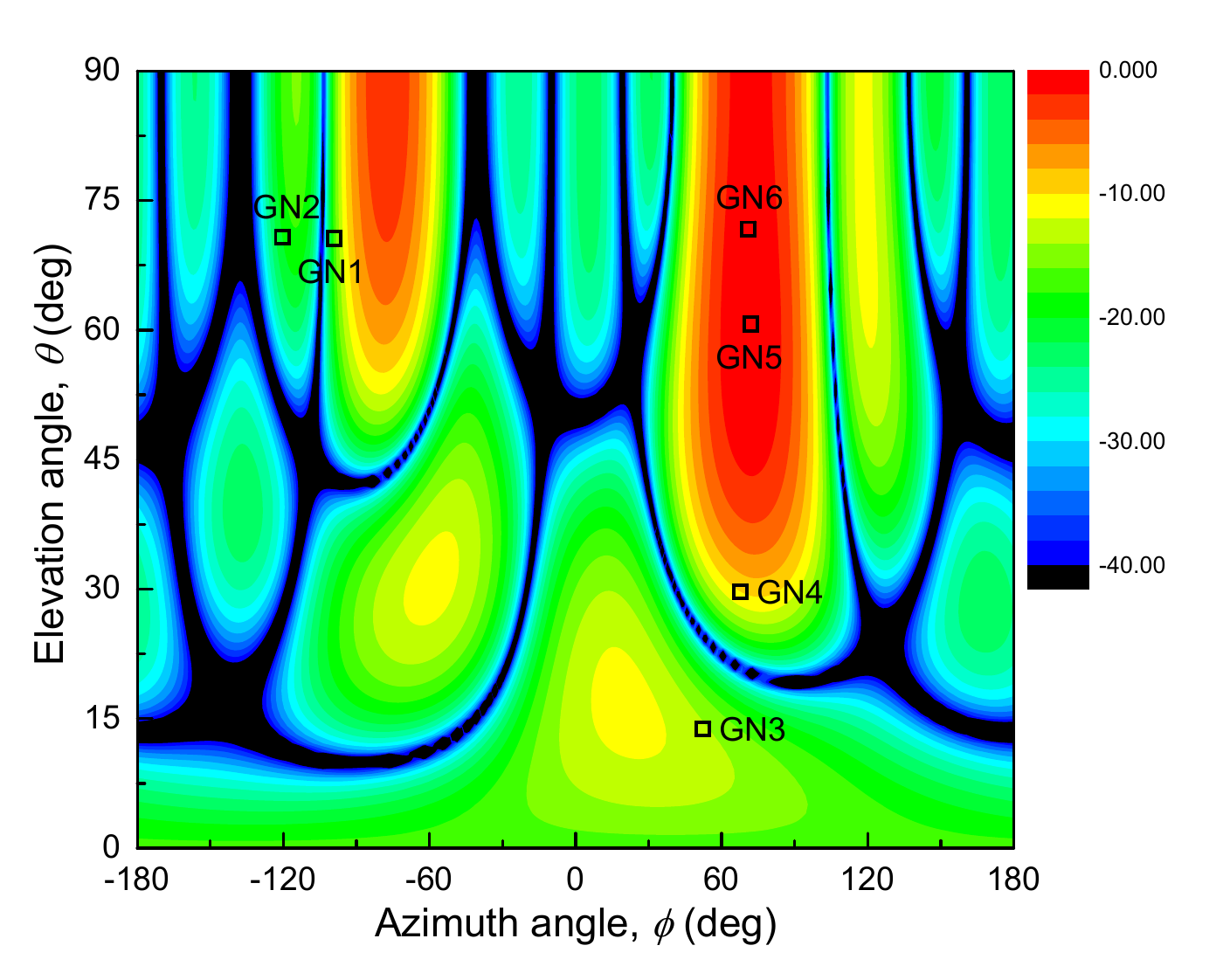}\label{R1-4}
    }
    \subfigure[Beam gain pattern for GN 5 after UAV positioning optimization.]{
      \includegraphics[width=0.315\linewidth]{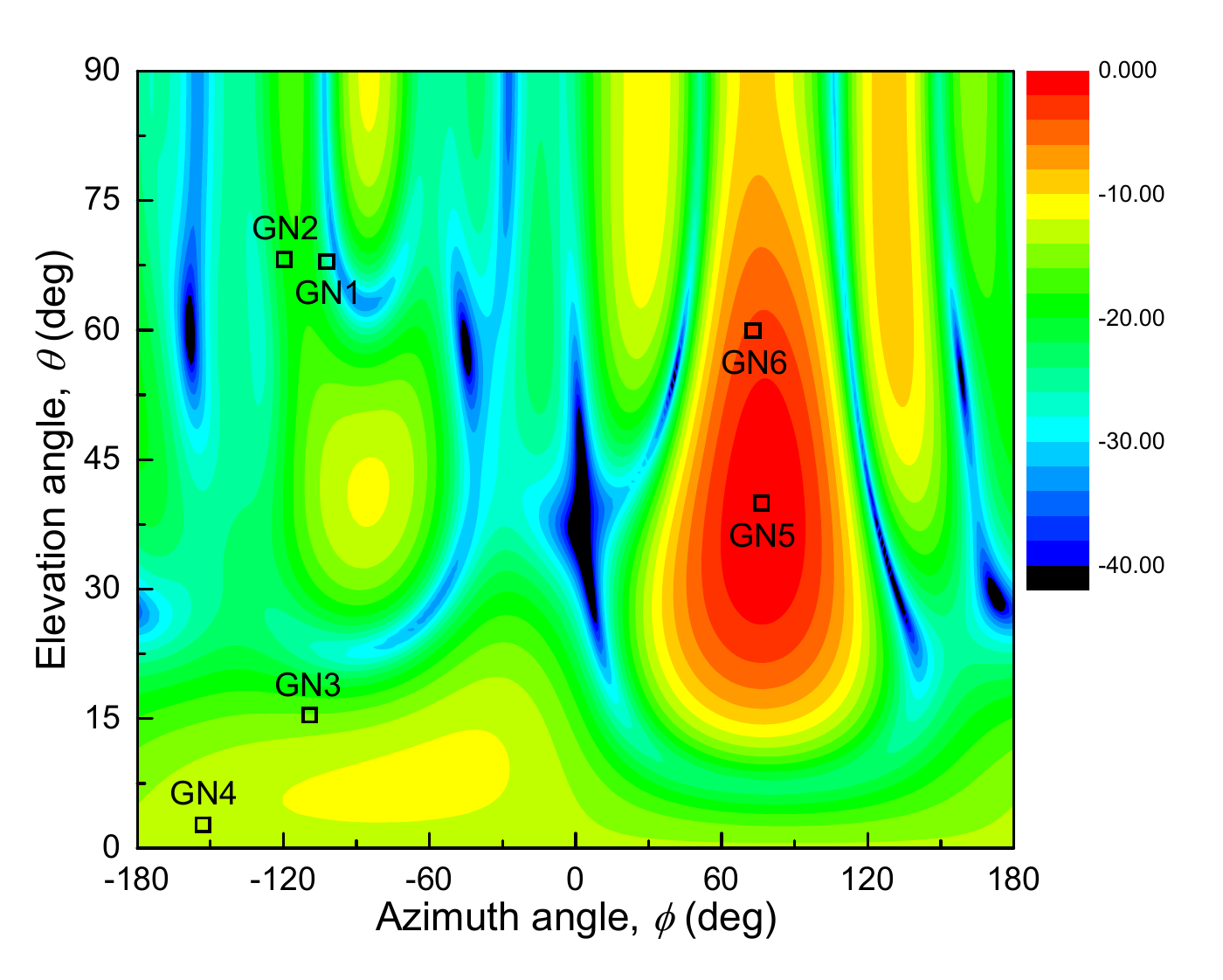}\label{R1-5}
    }
    \subfigure[Beam gain pattern for GN 5 after UAV positioning and beamforming optimization.]{
      \includegraphics[width=0.315\linewidth]{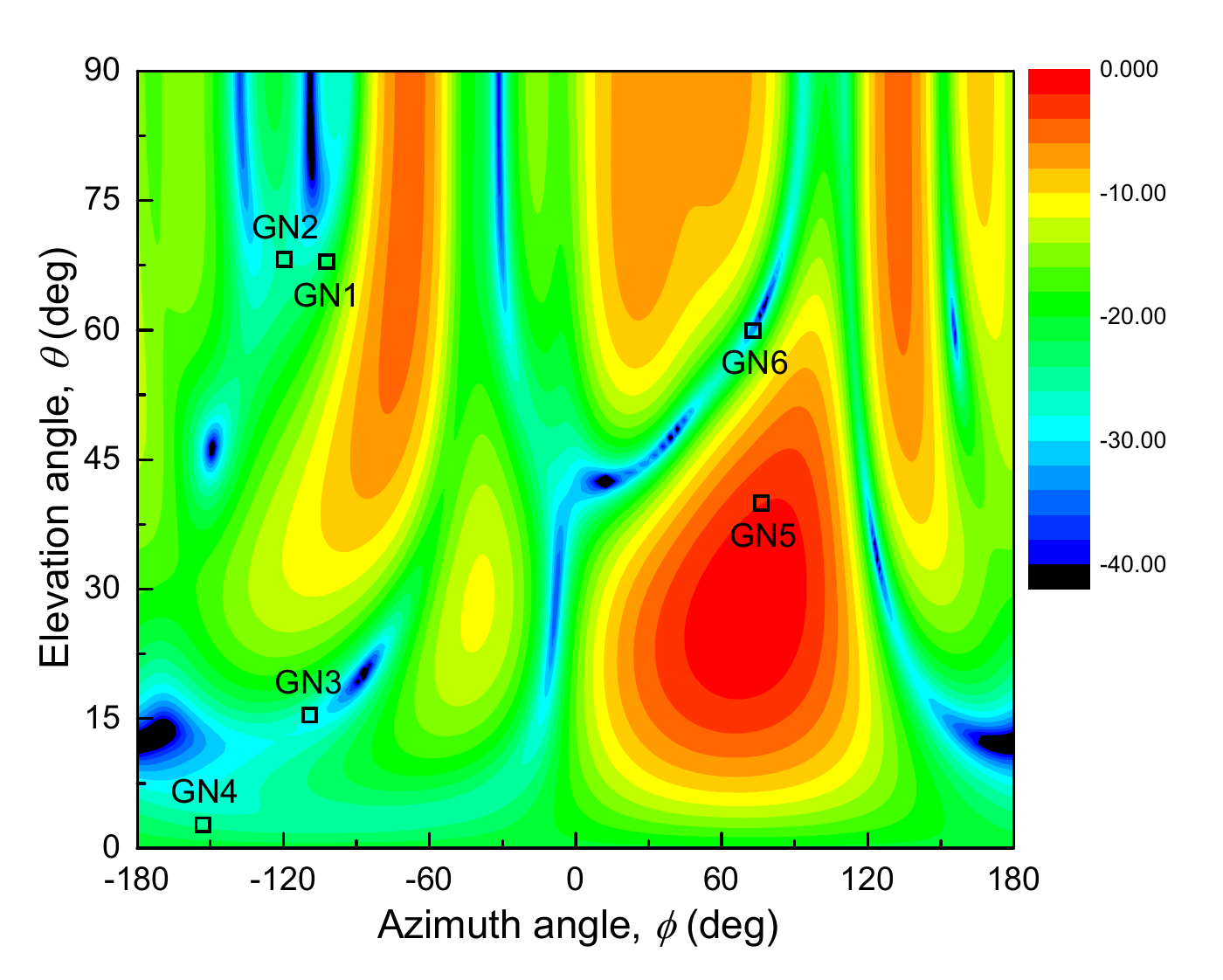}\label{R1-6}
    }
  \end{center} \vspace{-3mm}
\caption{Optimized UAV positions under different schemes, convergence behavior of the proposed scheme, and beam gain patterns for GN 5 at different optimization stages.} \vspace{-5mm}
\label{R1}
\end{figure*}

Fig. \ref{R1} illustrates the optimized UAV positions under different schemes, the convergence behavior of the proposed algorithm, and the beam gain patterns for a representative user at different optimization stages. For the given spatial distribution of GNs, it reveals the resource-allocation tendency of the proposed scheme, showing how the joint design adaptively adjusts the UAV position and beamforming pattern according to the underlying user geometry.

Fig. \ref{R1-1} marks the optimized UAV positions obtained by the baseline schemes and by an exhaustive search (ES), where the feasible region is uniformly discretized along each dimension and the beamforming matrix is optimized via WMMSE at every candidate location. For the proposed scheme, the full convergence path from the initial UAV position to the final solution is also illustrated. Both the proposed scheme and ES place the UAV at an altitude higher than the minimum allowable altitude, $z_{\textrm{min}}$ = 30 m, indicating that the optimal placement is determined not only by link-distance reduction but also by the need to improve spatial user separation and thereby mitigate inter-user interference. By contrast, TLAPB and IUP converge to nearly the same location, but for different reasons. IUP is mainly driven by desired-link enhancement and thus tends to approach a point that improves the link conditions of the most weakly served GNs, which in this geometry lies close to the centroid-based initial position. On the other hand, TLAPB updates the UAV position through local first-order Taylor approximation of the angle-dependent channel, which restricts the search to a neighborhood of the initial point and leads to convergence to a nearby suboptimal solution. Meanwhile, CFP moves away from the initial point toward a location that better separates GN 5 and GN 6, which exhibit the strongest steering correlation at the initial UAV position. This is because CFP fixes the inter-user steering-correlation term $\left|\mathbf{a}_k^{H}\mathbf{a}_i\right|^{2}$ at its initial value, so the positioning update is guided by the initially observed correlation pattern rather than the actual geometry-dependent correlation at the updated UAV position. In contrast, FCP keeps the UAV fixed at the center of the operating region.

Fig. \ref{R1-2} presents the minimum achievable SE as a heatmap over the horizontal plane. The proposed scheme and ES are both located in the darkest red region, indicating that they identify the most favorable UAV placement for max–min SE enhancement. CFP moves in a direction that mitigates the initially strong correlation between GN 5 and GN 6, but because the correlation term is fixed at the initial location, it cannot fully reflect the actual correlation changes caused by UAV repositioning, resulting in lower performance than the proposed scheme and ES. TLAPB and IUP choose relatively balanced locations while taking all GNs into account, but because they do not explicitly optimize the steering-correlation structure, they fail to reach the highest-performance region. FCP, which keeps the UAV fixed at the center of the operating region, shows the lowest performance.

Fig. \ref{R1-3} illustrates the convergence behavior of the proposed scheme by showing the minimum achievable SE and the corresponding maximum spatial correlation, defined as $\max_{k\neq i}\frac{\left|\mathbf{a}_k^{H}\mathbf{a}_i\right|^{2}}{\|\mathbf{a}_k\|^{2}\|\mathbf{a}_i\|^{2}}$, over the iterations. As the iterations proceed, the minimum achievable SE monotonically increases, whereas the maximum spatial correlation monotonically decreases. This indicates that the positioning stage progressively refines the UAV location to enlarge the angular separation among users and reduce the strongest normalized steering correlation. Once this more favorable spatial structure is established, the subsequent beamforming stage exploits instantaneous CSI at the optimized UAV location and yields a much sharper performance improvement by more effectively suppressing the remaining inter-user interference. The proposed scheme reduces the maximum spatial correlation to 0.509, while the benchmark schemes remain around 0.9, implying that they are much less effective in improving angular separability among users. Although ES further lowers this value to 0.467, it incurs a computational time of 45136 s. In contrast, the proposed scheme achieves a very similar UAV placement and a comparable level of maximum spatial correlation in only 123 s, showing that it can approximate the ES solution with substantially lower computation time.

\begin{figure*}[ht!]
  \begin{center}
    \subfigure[$\eta$ vs. $N_x$ = $N_y$.]{
      \includegraphics[width=0.365\linewidth]{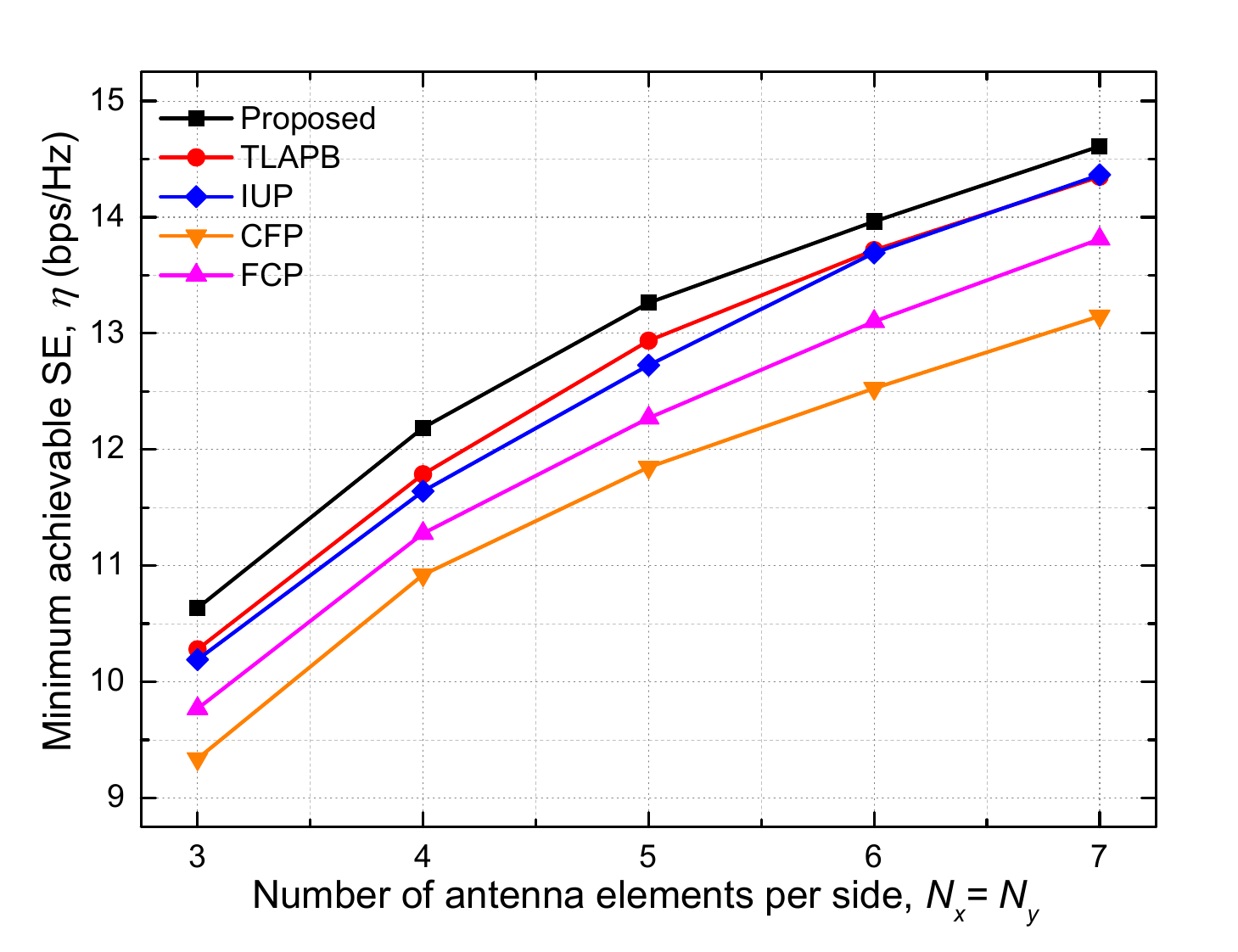}\label{R2-1}
    }
    \subfigure[$\eta$ vs. $K$.]{
      \includegraphics[width=0.365\linewidth]{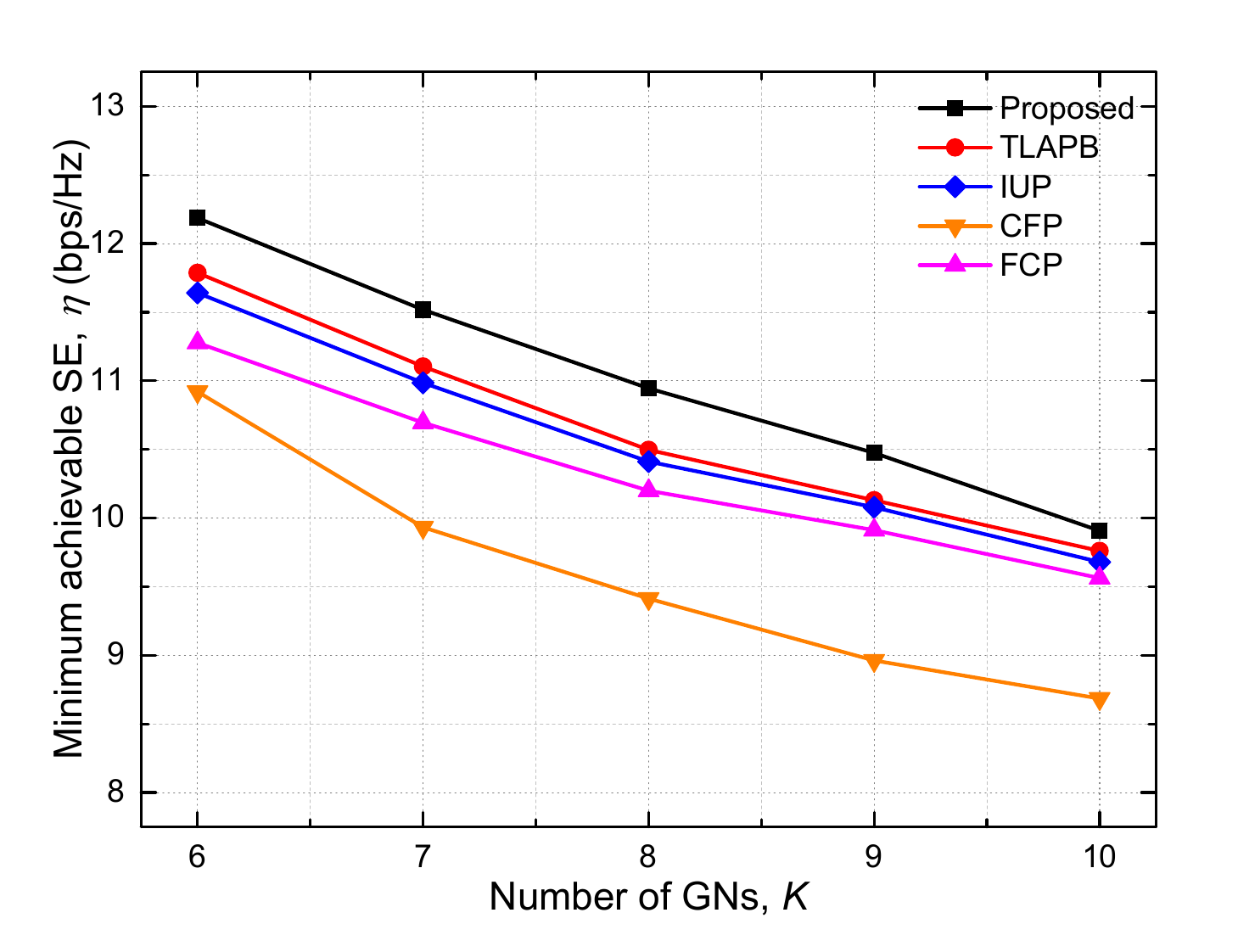}\label{R2-2}
    }
    \subfigure[$\eta$ vs. $K_R$.]{
      \includegraphics[width=0.365\linewidth]{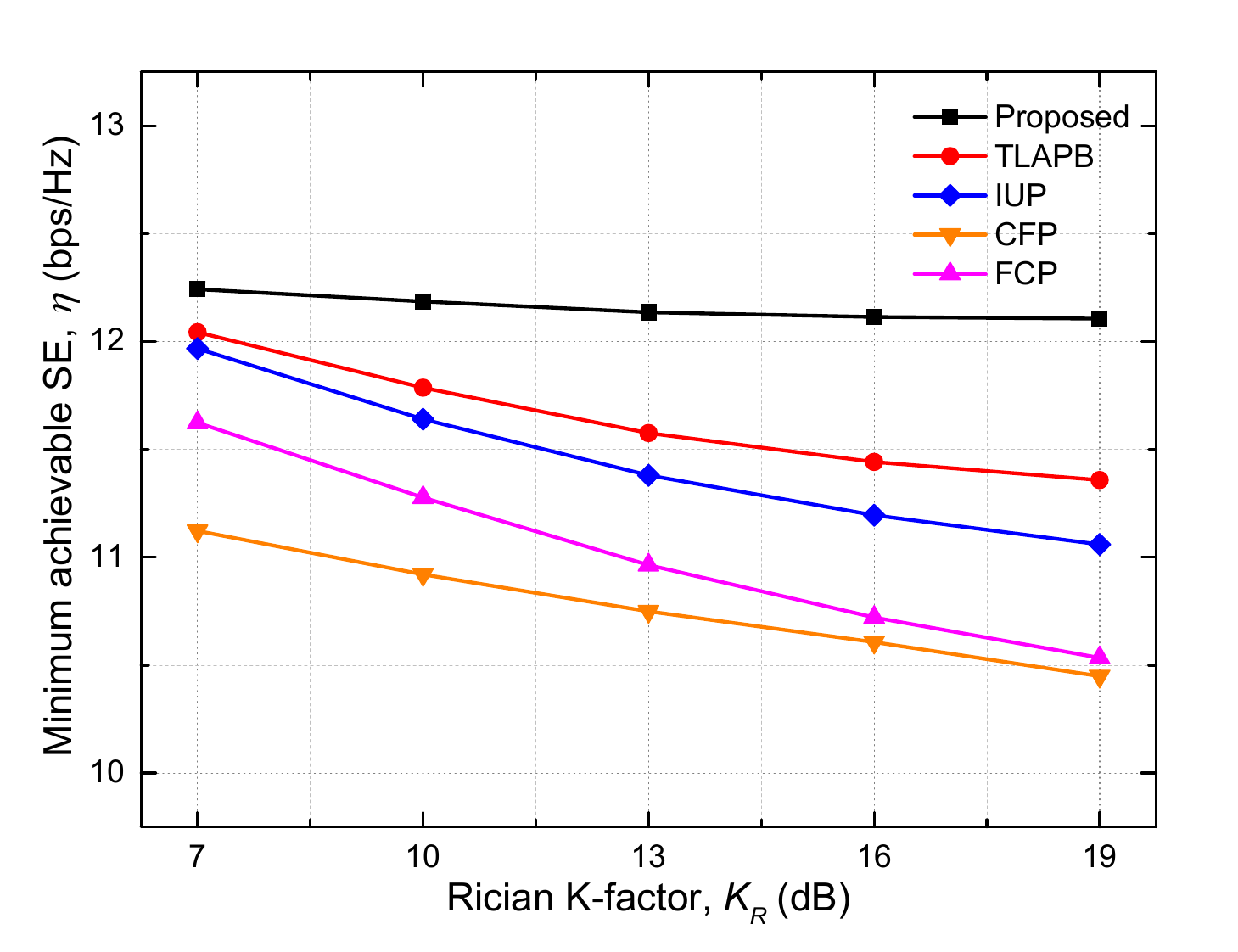}\label{R2-3}
    }
    \subfigure[$\eta$ vs. $P_{\textrm{max}}$.]{
      \includegraphics[width=0.365\linewidth]{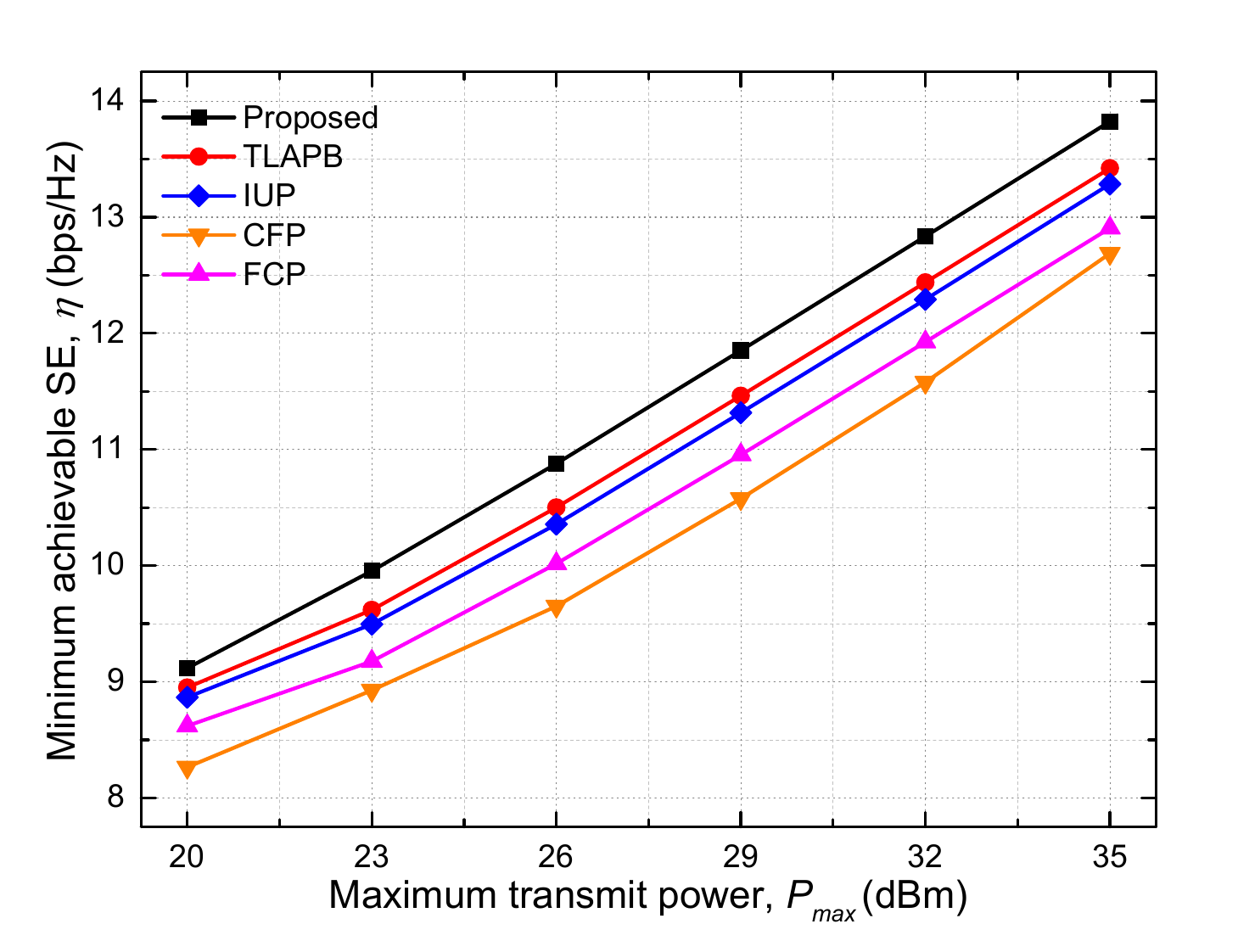}\label{R2-4}
    }
  \end{center} \vspace{-3mm}
\caption{Performance comparison under varying system parameters.} \vspace{-5mm}
\label{R2} 
\end{figure*}

Figs. \ref{R1-4}--\ref{R1-6} depict, for the proposed scheme, the normalized beam pattern of the beamformer associated with GN 5, which is the worst-performing user when the beamforming weights are initialized by MRT and the UAV is placed at the centroid of the GN distribution. Specifically, for each scanning angle pair $(\theta,\phi)$ with respect to the UAV, the beam gain is computed as $G_5(\theta,\phi)\!=\!|\mathbf{a}(\theta,\phi)^H\mathbf{w}_5|^2$, normalized by its maximum value, and plotted in dB. This beam pattern illustrates how strongly the beam assigned to GN 5 is formed in each direction, where a redder color indicates a stronger beam response. Hence, the figure allows us to examine whether a high gain is formed toward GN 5 while the leakage toward other GNs remains sufficiently suppressed. At the initial UAV position in Fig. \ref{R1-4}, both the directions of GN 5 and GN 6 appear in red. This indicates that, from the UAV’s perspective, the two GNs have similar angular directions, so that the beam intended for GN 5 also produces a strong response toward GN 6. As a result, significant leakage toward GN 6 is observed. After optimizing the UAV position in Fig. \ref{R1-5}, the angular separation among GNs is improved, and although GN 6 is not completely separated from GN 5, the response in the direction of GN 6 becomes visibly weaker than at the initial position. Finally, after the beamforming weights are further optimized in Fig. \ref{R1-6}, the gain toward GN 5 is maintained while the response toward GN 6 is further reduced, appearing as a much darker color. These results clearly show that UAV positioning first improves the spatial geometry by enhancing angular user separation, and the subsequent beamforming optimization then exploits this improved geometry to more effectively suppress the remaining inter-user leakage.

Fig. \ref{R2} compares the minimum achievable SE of the proposed scheme and the benchmark schemes under different system parameters, where each result is averaged over Monte Carlo realizations with randomly generated GN distributions and channel realizations. Fig. \ref{R2-1} shows that the proposed scheme achieves the best performance over the whole range of the number of antenna elements $N_x\!=\!N_y$, followed by TLAPB, IUP, FCP, and CFP in most cases. The performance gain comes from its ability to jointly exploit distance-dependent path-loss reduction and angle-dependent channel separability. Among the benchmark schemes, TLAPB generally provides the best performance because, starting from the centroid-based initial placement, its local first-order Taylor-based updates tend to remain in a relatively balanced region and thus converge to a nearby suboptimal solution that avoids excessively weak user links. However, TLAPB still fails to reach the globally more favorable region because it does not explicitly optimize the steering-correlation structure. IUP follows next, since its desired-signal-oriented positioning improves the link conditions of the worst-served GNs, although it does not explicitly account for the geometry-dependent interference structure. CFP performs the worst because it optimizes the UAV position using the steering-correlation values fixed at the initial location, and thus cannot capture the actual channel-separation improvement caused by UAV movement. Consequently, it may move the UAV toward a location biased to only a subset of GNs, leaving the remaining users as bottlenecks and even yielding lower performance than FCP. In contrast, FCP avoids such biased repositioning by simply fixing the UAV at the center of the operating region. Furthermore, the performance of all schemes increases with the number of antenna elements, since a larger array provides stronger beamforming gain and finer angular resolution. Consequently, both desired-signal enhancement and interference suppression become more effective. However, when the array size becomes sufficiently large, the additional benefit of geometry-aware UAV positioning becomes less pronounced because beamforming itself can already provide substantial spatial separation among users. In this regime, the performance bottleneck shifts from user separability to weakest-link enhancement, which may explain why the desired-signal-oriented placement of IUP slightly outperforms TLAPB at $N_x\!=\!N_y\!=\!7$.

In Fig. \ref{R2-2}, the performance of all schemes decreases as the number of GNs $K$ increases, because it becomes increasingly difficult for the UAV to provide both strong links and sufficient channel separation for a larger number of users. This degradation is more pronounced for CFP, since its frozen-correlation design tends to bias the UAV position toward only a subset of users, leaving the remaining users more likely to become bottlenecks as $K$ grows. In Fig. \ref{R2-3}, when the Rician $K$-factor is small, the stronger NLoS component makes channel separability less dependent on the LoS geometry, so the benefit of the proposed angle-aware UAV positioning is relatively less pronounced, and the performance gap among the schemes remains moderate. As the $K$-factor increases, the LoS component becomes more dominant, and the channel separability is increasingly determined by the angular geometry between the UAV and the GNs. Accordingly, the advantage of the proposed scheme becomes more pronounced, resulting in a larger performance gain over the benchmark schemes. In Fig. \ref{R2-4}, the performance of all schemes increases with the maximum transmit power $P_{\textrm{max}}$, since a larger power budget strengthens the desired signal level for all users. Overall, the proposed scheme consistently achieves the best performance over all tested parameter ranges by jointly accounting for distance-dependent path-loss effects and angle-dependent channel separability. This confirms the effectiveness of the proposed correlation-aware geometric design and its robustness under various system configurations.

\vspace{-2mm}
\section{Conclusions}

This paper studied joint UAV placement and beamforming for multiuser UAV-assisted downlink communications. Unlike existing designs in which UAV placement is not directly guided by the geometry-dependent interference induced by inter-user steering correlation, we developed a correlation-aware geometric framework that incorporates steering-vector correlation into the positioning stage. This enables the UAV location to be optimized with explicit consideration of inter-user interference structure. To make the resulting design tractable, we introduced a safeguarded Gaussian surrogate together with a tractable SE formulation for the positioning stage based on long-term channel statistics, and integrated them into a sequential two-stage optimization framework for UAV positioning and beamforming design. Simulation results showed that the proposed framework achieves notable gains in the minimum user SE over benchmark schemes by enlarging angular separation among users and reducing inter-user interference. More importantly, the results reveal a key design insight: in multiuser downlink systems, UAV positioning should be viewed not only as a means of improving link strength, but also as a geometric mechanism for mitigating multiuser interference through improved user separation. These findings highlight the broader role of UAV positioning in future interference-aware multi-antenna UAV communication systems.

\vspace{-2mm}

\end{document}